\documentclass[12pt, onecolumn, twoside]{IEEEtran}

\usepackage{graphicx}
\usepackage{amssymb}
\usepackage{amsmath}
\usepackage{dsfont}
\usepackage{url}
\usepackage{setspace}

\newtheorem{theorem}{Theorem}
\newtheorem{cor}{Corollary}
\newtheorem{lem}{Lemma}

\newtheorem{defn}{Definition}

\newcommand{\Reals}{\mathds R}

\newcommand{\bx}{\mathbf{x}}

\newcommand{\by}{\mathbf{y}}

\newcommand{\bu}{\mathbf{u}}

\newcommand{\bX}{\mathbf{X}}
\newcommand{\bY}{\mathbf{Y}}

\newcommand{\bE}{\mathbf{E}}

\newcommand{\cC}{\mathcal{C}}

\newcommand{\cI}{\mathcal{I}}

\newcommand{\EE}{\mathbb{E}}

\newcommand{\ra}{\rightarrow}

\newcommand{\VAR}{V\!AR}

\usepackage{psfrag}
\usepackage{color}
\usepackage{import}
\usepackage{bm}

\newcommand{\ID}{\mathrm{ID}}

\newcommand{\CAP}{\mathrm{CAP}}
\newcommand{\BALL}{\mathrm{BALL}}
\newcommand{\CONE}{\mathrm{CONE}}
\newcommand{\ST}{S^{\mathrm{typ}}}
\newcommand{\bU}{\mathbf{U}}
\newcommand{\ERASURE}{\mathtt{e}}

\newtheorem{remark}{Remark}

\newcommand{\MAYBE}{\mathtt{maybe}}

\newcommand{\NO}{\mathtt{no}}

\begin{document}

\title{Compression for Quadratic Similarity Queries}
\author{Amir Ingber, Thomas Courtade and Tsachy Weissman
\thanks{The material in this paper was presented in part at the 2013 Data Compression Conference (DCC), Snowbird, UT.}
\thanks{The authors are with the Dept. of Electrical Engineering, Stanford University, Stanford, CA 94305.
Email: \{ingber, courtade, tsachy\}@stanford.edu.}
\thanks{This work is supported in part by the NSF Center for Science of Information
under grant agreement CCF-0939370.}
}

\maketitle
\markboth{Submitted to IEEE Transactions on Information Theory}{Ingber \MakeLowercase{et al.}: Compression for Quadratic Similarity Queries}

\begin{abstract}
The problem of performing similarity queries on compressed data is considered.
We focus on the quadratic similarity measure, and study the fundamental tradeoff between compression rate, sequence length, and  reliability of queries performed on compressed data.
For a Gaussian source, we show that queries can be answered reliably if and only if the compression rate exceeds a given threshold -- the \emph{identification rate} -- which we explicitly characterize.
Moreover, when compression is performed at a rate greater than the identification rate, responses to queries on the compressed data can be made exponentially reliable.
We give a complete characterization of this exponent, which is analogous to the error and excess-distortion exponents in channel and source coding, respectively.

For a general source we prove that, as with classical compression, the Gaussian source requires the largest compression rate among sources with a given variance.
Moreover, a robust scheme is described that attains this maximal rate for \emph{any source distribution}.
\end{abstract}


\section{Introduction}
For a database consisting of many long sequences, it is natural to perform queries of the form:
\textit{which sequences in the database are similar to a given sequence} $\by$?
In this paper, we study the problem of compressing this database so that queries about the original data can be answered reliably given only the compressed version.
This goal stands in contrast to the traditional compression paradigm, where data is compressed so that it can be reconstructed -- either exactly or approximately -- from its compressed form.

Specifically, for each sequence $\bx$ in the database we only keep a short \emph{signature}, denoted  $T(\bx)$, where $T(\cdot)$ is a  signature assignment function.
Queries are performed using only $\by$ and $T(\bx)$ as input, rather than the original (uncompressed) sequence $\bx$. This setting is illustrated in Fig.~\ref{fig:query}.

\begin{figure}
  \centering
  \psfrag{&x1}{$\mathbf{x}_1$}
  \psfrag{&x2}{$\mathbf{x}_2$}
  \psfrag{&xM}{$\mathbf{x}_M$}
  \psfrag{&vd}{$\vdots$}
  \psfrag{&T}{$\!T(\cdot)$}
  \psfrag{&t1}{$\!t_1$}
  \psfrag{&t2}{$\!t_2$}
  \psfrag{&tM}{$\!t_M$}
  \psfrag{&y}{$\mathbf{y}$}
  \psfrag{&YN}{$\mathtt{yes/no}$}
  \psfrag{&dx1y}{\hspace{-.27in} is $\mathbf{x}_1 \cong \mathbf{y}$ ?}
  \psfrag{&dx2y}{\hspace{-.27in} is $\mathbf{x}_2 \cong \mathbf{y}$ ?}
  \psfrag{&dxMy}{\hspace{-.27in} is $\mathbf{x}_M \cong \mathbf{y}$ ?}

  \includegraphics[width=.8\textwidth]{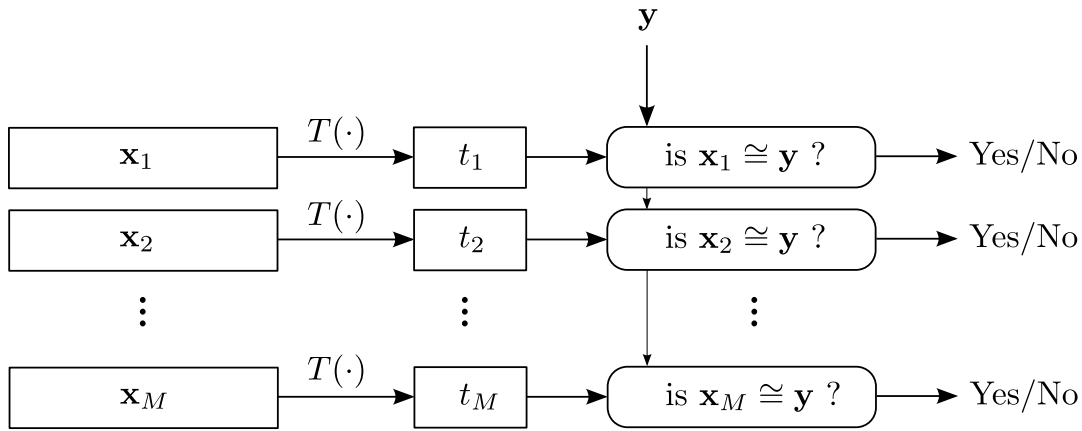}\\
  \caption{Answering a query from compressed data}\label{fig:query}
\end{figure}

As alluded to above, we generally do not require that the original data be reproducible from the  signatures.
Therefore the set of signatures is not meant to replace the database itself.
Nevertheless, there are many instances where such compression is desirable.
For example, the set of signatures can be thought of as a cached version of the original database (possibly hosted at many locations due to its relatively small size).
By performing queries only on the cached (i.e., compressed) database,  query latency can be reduced and the computational burden on the server hosting the uncompressed database can be lessened.

In many scenarios (e.g,. querying a criminal forensic database), query responses which  are false negatives are not acceptable.
A false negative occurs if a query performed on $T(\bx)$ and $\by$ indicates that $\bx$ and $\by$ are not similar, but they are in truth.
Therefore, we impose the restriction in our model that \emph{false negatives are not permitted}.
With this in mind, we regard the query responses from the compressed data as either ``$\NO$" or ``$\MAYBE$".
Since minimizing the probability that a query returns $\MAYBE$ is equivalent to minimizing the probability of returning a false
positive\footnote{Complementary to false negatives, a false positive occurs if a query performed on $T(\bx)$ and $\by$ indicates that $\bx$ and $\by$ are similar (i.e., returns $\MAYBE$), but they are not in truth.},
any good compression scheme will have a corresponding query function which returns $\MAYBE$ with small probability.
We note briefly that a false positive does not cause an error \emph{per se}.  Rather, it only introduces a computational burden due to the need for further verification.

In our setting we assume that the query and database sequences are independent from one another,  and all entries are drawn i.i.d.~according to a given distribution. The setting is closely related to the problem considered by Ahlswede et al. \cite{Ahlswede97}, where the focus was only on discrete sources. In \cite{Ahlswede97}, the authors attempt to attack the more general problem where both false positives and false negatives are allowed. In this general case, it was demonstrated in \cite{Ahlswede97} that the question of `achievable rate' is uninteresting and only the error exponent is studied (in the current paper, where false negatives are not allowed, we show that the rate question becomes interesting again). We should also note that the error exponent results in \cite{Ahlswede97} are parameterized by an auxiliary random variable with unbounded alphabet cardinality, rendering those quantities incomputable, and therefore of limited practical interest.
Another closely related work is the one by Tuncel et al. \cite{tuncel2004rate}, where the search accuracy was addressed by a reconstruction requirement with a single-letter distortion measure that is side-information dependent
(and the tradeoff between compression and accuracy is that of a Wyner-Ziv \cite{WynerZiv} type).
In contrast, in the current paper the search accuracy is measured directly by the accuracy of the query answers.

A different line of work attempting to identify the fundamental performance limits of database retrieval includes \cite{OSullivan2002,Willems03},
which characterized the maximum rate of entries that can be reliably identified in a database.
This line of work was extended independently in \cite{Westover08,Tuncel09} allowing compression of the database,
and in \cite{Tuncel12submitted} to the case where sequence reconstruction is also required.
In each of these works, the underlying assumption is that the original sequences are corrupted by noise before their enrollment in the database,
the query sequence is \emph{one of those original sequences}, and the objective is to identify which one.
There are two fundamental differences between this line of work and ours.
First, in our case the query sequence is random (i.e. generated by nature) and does not need to be a sequence that has already been enrolled in the database.
Second, in our problem we  attempt to identify  sequences that are \emph{similar} to the query sequence, rather than an exact match.

Other related ideas in the literature include Bloom Filters \cite{Bloom1970} (with many subsequent improvements, e.g. \cite{Porat09_Opt_Bloom_Filter_matrix}),
which are efficient data structures enabling queries without false negatives.
The Bloom Filter only applies for exact matches (where here we are interested in similarity queries) so it is not applicable to our problem.
Nevertheless, as surveyed in \cite{MitzenmacherBloomFilter}, Bloom filters demonstrate the potential of answering queries from compressed data.

Another related notion is that of Locality Sensitive Hashing (LSH), which is a framework for data structures and algorithms for finding similar items in a given set
(see \cite{AndoniI08} for a survey).
LSH trades off accuracy with computational complexity and space, and false negatives are allowed.
Two fundamental points are different in our approach.
First, we study the information-theoretic aspect of the problem, i.e.,  we concentrate on space only (compression rate) and ignore computational complexity in an attempt to understand the amount of information {relevant to querying} that can be stored in the short signatures.
Second, we do not allow false negatives, which, as discussed above, are inherent for LSH.

Other approaches for similarity search from compressed data involve dimensionality reduction techniques that preserve distances, namely those based on Johnson-Lindenstrauss-type embeddings \cite{JohnsonLindenstrauss84} (see also \emph{sketching}, e.g. \cite{sketchingNotes}).
A recent interesting application of this approach involves image retrieval for an augmented reality setting \cite{BR_DCC13}. However, note that such mappings generally depend on the elements in the database; the distance preservation property cannot apply to \emph{any} query element outside the database, making the guarantee for zero false negatives impossible without further assumptions.

\bigskip
This paper is organized as follows. In the next section we formally define the problem and the quantities we study (i.e., the identification rate and the identification exponent).
In Section~\ref{sec:results} we state and discuss our main results. Section~\ref{sec:proofs} provides the proofs of these results, and Section~\ref{sec:summary} delivers concluding remarks.

\section{Problem Formulation}
Throughout this paper, boldface notation $\bx$ denotes a column vector of elements $[x_1,...x_n]^T$. Capital letters denote random variables (e.g. $X,Y$), and $\bX,\bY$ denote random vectors.  Throughout the paper $\log(\cdot)$ denotes the base-$2$ logarithm, while $\ln(\cdot)$ is used for the usual natural logarithm.

We focus on the basic notion of quadratic similarity (sometimes called mean square error, or MSE). To this end, for any length-$n$ real sequences $\bx$ and $\by$ define
\begin{align}
  d(\bx,\by) \triangleq \frac{1}{n}\sum_{i=1}^n (x_i-y_i)^2 = \frac{1}{n}\|\bx-\by\|^2,\label{eqn:mse}
\end{align}
where $\|\cdot\|$ denotes the standard Euclidean norm. We say that $\bx$ and $\by$ are $D$-\emph{similar} when $d(\bx,\by)\leq D$, or simply \emph{similar} when $D$ is clear from  context.

A rate-$R$ identification system $(T,g)$ consists of a \emph{signature assignment}
\begin{align}
T : \Reals^n \ra \{1,2,\dots,2^{nR}\}
\end{align}
and a \emph{query function}
\begin{align}
g : \{1,2,\dots,2^{nR}\}\times \Reals^n \ra \{\NO, \MAYBE\}.
\end{align}

A system $(T,g)$ is said to be $D$-\emph{admissible}, if for any $\bx,\by$ satisfying $d(\bx,\by)\leq D$, we have
\begin{equation}\label{eqn:maybe}
  g(T(\bx),\by) = \MAYBE.
\end{equation}
This notion of $D$-{admissibility} motivates the use of ``$\NO$" and ``$\MAYBE$" in describing the output of $g$:
\begin{itemize}
\item If $g(T(\bx),\by) = \NO$, then $\bx$ and $\by$ can not be $D$-similar.
\item If $g(T(\bx),\by) = \MAYBE$, then $\bx$ and $\by$ are possibly $D$-similar.
\end{itemize}
Stated another way, a $D$-{admissible} system $(T,g)$ does not produce false negatives, i.e., indicate that $\bx$ and $\by$ are not similar, when they are in truth.  Thus, a natural figure of merit for a $D$-{admissible} system $(T,g)$ is the frequency at which false positives occur (i.e., where $g(T(\bx),\by) = \MAYBE$ and $d(\bx,\by)>D$).  To this end, let $P_X$ and $P_Y$ be probability distributions on $\Reals$, and assume  $(\bX,\bY)\sim \prod_{i=1}^n P_X(x_i)P_Y(y_i)$.  That is, the vectors $\bX$ and $\bY$ are independent of each other and drawn i.i.d.\ according to $P_X$ and $P_Y$ respectively.  Define the \emph{false positive event}
\begin{align}
\mathcal{E} = \{ g(T(\bX),\bY) = \MAYBE,  d(\bX,\bY) > D\}, \label{FPeventDefn}
\end{align}
and note that, for any $D$-admissible system $(T,g)$, we have
\begin{align}
\Pr \{ g(T(\bX),\bY) = \MAYBE \}
&= \Pr \{ g(T(\bX),\bY) = \MAYBE | d(\bX,\bY)\leq D\}\Pr\{ d(\bX,\bY)\leq D\} \notag\\
&\quad+ \Pr\{ g(T(\bX),\bY) = \MAYBE,  d(\bX,\bY) > D\} \\
&=\Pr\{ d(\bX,\bY)\leq D\} + \Pr\{\mathcal{E}\}, \label{eqn:FP_maybeRelation}
\end{align}
where \eqref{eqn:FP_maybeRelation} follows since $\Pr \{ g(T(\bX),\bY) = \MAYBE | d(\bX,\bY)\leq D\}=1$ by $D$-admissibility of $(T,g)$.  Since $\Pr\{ d(\bX,\bY)\leq D\}$ does not depend on what scheme is employed, minimizing the false positive probability $\Pr\{\mathcal{E}\}$ over all $D$-admissible schemes $(T,g)$ is equivalent to minimizing $\Pr \{ g(T(\bX),\bY) = \MAYBE \}$. Also note, that the only interesting case is when $\Pr\{d(\bX,\bY)\leq D \} \ra 0$ as $n$ grows, since otherwise almost all the sequences in the database will be similar to the query sequence, making the problem degenerate (since almost all the database needs to be retrieved, regardless of the compression). In this case, it is easy to see that $\Pr\{\mathcal{E}\}$ vanishes if and only if the conditional probability
\begin{equation}
  \Pr \{ g(T(\bX),\bY) = \MAYBE | d(\bX,\bY)> D\}
\end{equation}
vanishes as well.
In view of the above, we henceforth restrict our attention to the behavior of $\Pr \{ g(T(\bX),\bY) = \MAYBE \}$.  In particular, we study the tradeoff between the rate $R$ and $\Pr \{ g(T(\bX),\bY) = \MAYBE \}$.

This motivates the following definitions:
\begin{defn}
    For given distributions $P_X, P_Y$ and a similarity threshold $D$, a rate $R$ is said to be $D$-\emph{achievable} if there exists a sequence of  rate-$R$ admissible schemes $(T^{(n)},g^{(n)})$ satisfying
    \begin{equation}
      \lim_{n\ra\infty} \Pr\left\{g^{(n)}\left(T^{(n)}(\bX),\bY \right) = \MAYBE\right\} = 0. \label{eqn:reliable}
    \end{equation}
\end{defn}
\begin{defn}
    For given distributions $P_X, P_Y$ and a similarity threshold $D$, the \emph{identification rate}  $R_\ID(D,P_X,P_Y)$ is the infimum of $D$-achievable rates.  That is,
   \begin{align}
     R_\ID(D,P_X,P_Y) \triangleq \inf \{ R : R~ \mbox{is $D$-achievable}\},
   \end{align}
   where an infimum over the empty set is equal to $\infty$.
\end{defn}

 The above definitions are in the same spirit of the rate distortion function (the rate above which a vanishing probability for excess distortion is achievable),
and also in the spirit of the channel capacity (the rate below which a vanishing probability of error can be obtained).
See, for example, Gallager~\cite{GallagerInfoTheoryBook}%
.\footnote{See, for example, Cover and Thomas \cite{CoverThomas_InfoTheoryBook} for the alternative approach based on average distortion rather than excess distortion probability.}

Having defined $R_\ID(D,P_X, P_Y)$, the rate at which $\Pr \{ g(T(\bX),\bY) = \MAYBE \}$ vanishes is also of significant interest.
We expect the vanishing rate to be exponential as in the traditional source coding setting, motivating the following definition:
\begin{defn}
    Fix  $R\geq R_\ID(D,P_X,P_Y)$. The \emph{identification exponent} is defined as 
    \begin{align}
        \bE_\ID(R,D,P_X,P_Y) \triangleq \limsup_{n \ra \infty} -\frac{1}{n}\log \inf_{g^{(n)},T^{(n)}} \Pr \left\{ g^{(n)}\left(T^{(n)}(\bX),\bY \right) = \MAYBE \right\},\label{eqn:ExponentDef}
    \end{align}
    where the infimum is over all $D$-admissible systems $(g^{(n)},T^{(n)})$ of rate $R$ and blocklength $n$.
\end{defn}

The  analogous quantity in  source coding is the excess distortion exponent, first studied by Marton \cite{Marton1974fidelityCriterion}
for discrete sources and by Ihara and Kubo \cite{IharaKubo2000} for the Gaussian source (see also \cite{IharaKubo2005} and \cite{zhongAC2006Laplacian} for other sources).

We pause to make a few additional remarks on the connection between $\Pr \{ g(T(\bX),\bY) = \MAYBE \}$ and $\Pr \{ \mathcal{E} \}$, where $\mathcal{E}$ is the false positive event defined in \eqref{FPeventDefn}.  If $P_X$ and $P_Y$ have identical means and finite variances $\sigma_X^2$ and $\sigma_Y^2$, respectively, then the weak law of large numbers implies
\begin{align}
\lim_{n\ra \infty}\Pr\{ d(\bX,\bY)\leq D\} = 0
\end{align}
when $D < \sigma_X^2+\sigma_Y^2$.  Thus, the relation \eqref{eqn:FP_maybeRelation} implies that vanishing $\Pr \{ \mathcal{E} \}$ is attainable if and only if $R>R_\ID(D,P_X,P_Y)$ when $D < \sigma_X^2+\sigma_Y^2$.
Finally, observe that \eqref{eqn:FP_maybeRelation} implies the relationship
\begin{align}
 &\bE_\ID(R,D,P_X,P_Y)  \notag \\
 &= \limsup_{n \ra \infty} -\frac{1}{n}\log \max \left[  \Pr\{ d(\bX,\bY)\leq D\}  , \inf_{g^{(n)},T^{(n)}} \Pr \left\{ \mathcal{E}^{(n)} \right\}  \right],  \label{FPandMaybeRelation}
\end{align}
where $\mathcal{E}^{(n)}$ is the false positive event defined via \eqref{FPeventDefn} for the system $(g^{(n)},T^{(n)})$, and the infimum is taken over all $D$-admissible systems $(g^{(n)},T^{(n)})$ of rate $R$ and blocklength $n$.

\section{Main Results}\label{sec:results}
This section delivers our main results; all proofs are given in Section \ref{sec:proofs}.  The Gaussian distribution plays a prominent role in this section, therefore we use the shorthand notation $P_X = N(\mu,\sigma^2)$ to denote that  $P_X$ is the Gaussian distribution on $\Reals$ with mean $\mu$ and variance $\sigma^2$.

\subsection{The Identification Rate for Gaussian Sources}

\begin{theorem}\label{thm:RIDdiffVar}
If $P_X=N(\mu,\sigma_X^2)$ and $P_Y=N(\mu,\sigma_Y^2)$, then
\begin{align}
  R_\ID(D,P_X,P_Y) =
  \left\{
      \begin{array}{ll}
     0 &\mbox{for $0 \leq D < (\sigma_X - \sigma_Y)^2$}\\
     \log\frac{2\sigma_X\sigma_Y}{\sigma_X^2+\sigma_Y^2-D} \quad  &\mbox{for $(\sigma_X - \sigma_Y)^2 \leq D < \sigma_X^2 + \sigma_Y^2$}\\
       \infty & \mbox{for $D \geq \sigma_X^2 + \sigma_Y^2$}.
      \end{array}
       \right. \label{eqn:RIDdiffVar}
\end{align}
\end{theorem}

Before proceeding, we make a few observations about the behavior of $R_\ID(D,P_X,P_Y)$ under the assumptions of Theorem \ref{thm:RIDdiffVar}.  First, the fact that $R_\ID(D,P_X,P_Y)=\infty$ for $D \geq \sigma_X^2 + \sigma_Y^2$ is not surprising.  Indeed, if $D \geq \sigma_X^2 + \sigma_Y^2$, then $\bX$ and $\bY$ are inherently $D$-similar.  That is, $\Pr\{ d(\bX,\bY)\leq D\}$ is bounded away from zero (it actually converges to $1$), and therefore \eqref{FPandMaybeRelation} reveals that $\Pr \{ g(T(\bX),\bY) = \MAYBE \}$ can never vanish, regardless of what scheme is used.  Second, \eqref{eqn:RIDdiffVar} is symmetric with respect to $\sigma_X^2$ and $\sigma_Y^2$.  Though this might be expected, it is not obviously true from the outset.  Finally, for fixed $\sigma_X^2$ and $D<\sigma_X^2$, the function $R_\ID(D,P_X,P_Y)$ given by \eqref{eqn:RIDdiffVar} is maximized when $\sigma_Y^2 = \sigma_X^2 - D$.  In Fig.~\ref{fig:RIDdiffVar} we plot \eqref{eqn:RIDdiffVar} for different values of $\sigma_Y^2$ in order to illustrate some of its properties.
\begin{figure}
  \centering
  \includegraphics[width=6in]{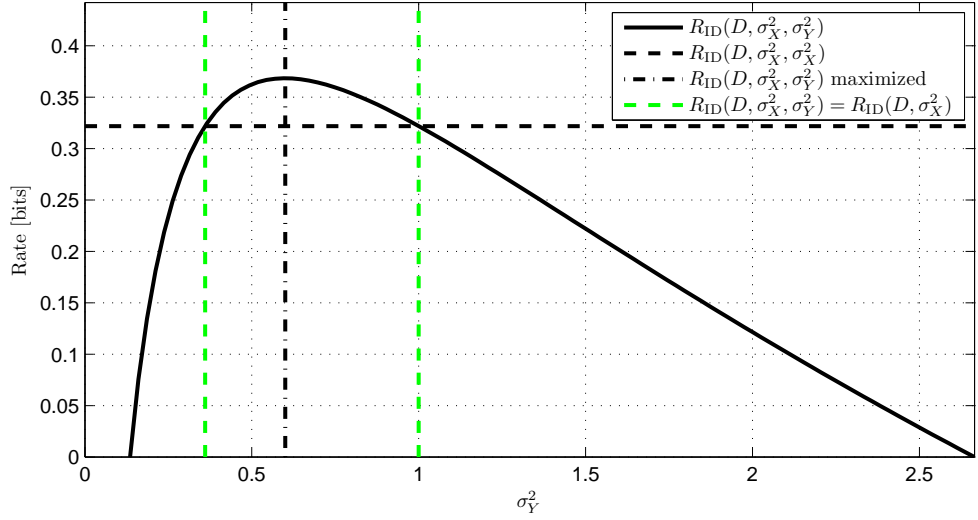}\\
  \caption{The identification rate $R_\ID(D,\sigma_X^2,\sigma_Y^2) := R_\ID(D,N(\mu,\sigma_X^2),N(\mu,\sigma_Y^2))$ for different values of $\sigma_Y^2$. Here $\sigma_X^2= 1$ and $D=0.4$.}\label{fig:RIDdiffVar}
\end{figure}

As an immediate corollary to Theorem \ref{thm:RIDdiffVar}, we obtain the following concise result for the symmetric case of $P_X=P_Y=N(\mu,\sigma^2)$.
\begin{cor}
If $P_X=P_Y=N(\mu,\sigma^2)$, then
    \begin{align}
      R_\ID(D,P_X,P_Y) = \left\{
      \begin{array}{ll}
      \log\left(\frac{2\sigma^2}{2\sigma^2-D}\right)\quad \quad &\mbox{for $0 \leq D < 2\sigma^2$}\\
       \infty & \mbox{for $D \geq 2\sigma^2$}.
      \end{array}
       \right.
       \label{eqn:RID(D)}
    \end{align}
\end{cor}

We remark that \eqref{eqn:RID(D)} is reminiscent of the Gaussian rate distortion function $R(D) = \left[\frac{1}{2}\log\frac{\sigma^2}{D}\right]^+$ (cf. \cite{CoverThomas_InfoTheoryBook}). The identification rate $R_\ID(D,N(\mu,\sigma^2),N(\mu,\sigma^2))$ and rate distortion function $R(D)$ for a Gaussian source are plotted in Fig.~\ref{fig:RD_RID}, and as
\begin{figure}
  \centering
  \includegraphics[width=.9\textwidth]{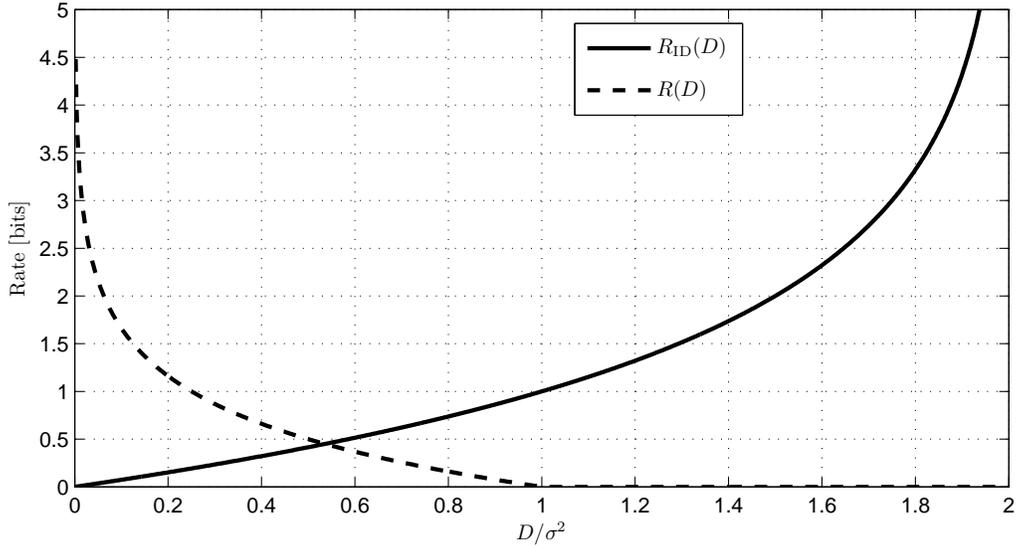}\\
  \caption{The identification rate $R_\ID(D) := R_\ID(D,N(\mu,\sigma^2),N(\mu,\sigma^2))$ and the rate distortion function $R(D)$ for a Gaussian source with variance $\sigma^2$.}\label{fig:RD_RID}
\end{figure}
seen in the figure, $R(D)$ is monotonically \emph{decreasing} in $D$, while \eqref{eqn:RID(D)}
is monotone \emph{increasing}. This can be intuitively explained by thinking of the compression scheme as a quantizer, where all the $\bx$ sequences mapped to the same $i\in\{1,2,\dots,2^{nR}\}$ define a quantization cell. Since the scheme must answer $\MAYBE$ for all sequences $\by$ similar to $\bx$, it therefore has to answer $\MAYBE$ for all $\by$ in the $D$-expansion of the quantization cell (all sequences that are at distance $D$ from \emph{any} point in the cell). The probability of $\MAYBE$ is, therefore, the probability that $\bY$ falls in the expanded cell, and this probability increases as either $D$ grows, or as the size of the quantization cell itself grows (i.e. the rate decreases).

\subsection{The Identification Exponent for Gaussian Sources}
Having established the identification rate for Gaussian sources, we now turn our attention to the identification exponent.
In order to simplify the notation for the identification exponents, we define the following functions
\begin{align}
  \bE_Z(\rho) &\triangleq \frac{1}{2\ln 2} (\rho - 1 -\ln \rho) \label{eqn:EZ}\\
\wp(R,D,z_1,z_2) &\triangleq -\log\sin \min \left[\tfrac{\pi}{2},\left(\arcsin\left(2^{-R}\right) + \arccos\frac{z_1+z_2 - D}{2\sqrt{z_1 z_2}} \right)\right].
\end{align}

\begin{theorem}\label{thm:EIDdiffVar}
Let $P_X=N(\mu,\sigma_X^2)$ and $P_Y=N(\mu,\sigma_Y^2)$.  For any fixed rate $R>R_\ID(D,P_X,P_Y)$,
\begin{align}
    \bE_\ID(R,D,P_X,P_Y) = &\min_{\rho_X,\rho_Y} \  \bE_Z(\rho_X)+\bE_Z(\rho_Y) +\wp(R,D,{\rho_X\sigma_X^2},{\rho_Y\sigma_Y^2}), \label{eqn:EIDdiffVar}
\end{align}
where the minimization is over all 
 $\rho_X,\rho_Y >0$ satisfying
\begin{align}
  &\left|\sqrt{\rho_X\sigma_X^2}-\sqrt{\rho_Y\sigma_Y^2}\right| < \sqrt{D},
  &\rho_X\sigma_X^2+\rho_Y\sigma_Y^2 \geq D. \label{eqn:rhoConditions}
\end{align}
\end{theorem}

\begin{remark}
We note that, for $P_X=N(\mu,\sigma_X^2)$ and $P_Y=N(\mu,\sigma_Y^2)$,  the exponent $\bE_\ID(R,D,P_X,P_Y)$ is strictly positive for $R>R_\ID(D,P_X,P_Y)$, and is equal to zero at $R = R_\ID(D,P_X,P_Y)$. Therefore, the direct part of Theorem~\ref{thm:RIDdiffVar} is implied by Theorem~\ref{thm:EIDdiffVar}.  However, the converse part of Theorem~\ref{thm:RIDdiffVar} is \emph{not} implied by Theorem~\ref{thm:EIDdiffVar},
as the latter does not exclude the possibility that the probability of $\MAYBE$ can be made to vanish with a sub-exponential decay rate when the exponent is equal to zero.
\end{remark}

In light of Theorem \ref{thm:EIDdiffVar}, it is instructive to revisit the relationship between false-positive and $\MAYBE$ probabilities specified in \eqref{FPandMaybeRelation}. To this end, consider the setting where $P_X=N(\mu,\sigma_X^2)$, $P_Y=N(\mu,\sigma_Y^2)$, and $D\leq \sigma_X^2+\sigma_Y^2$.  In this case, the random variable $\frac{1}{n(\sigma_X^2+\sigma_Y^2)}\|\bX-\bY\|^2$ has a chi-squared distribution  with $n$ degrees of freedom.  Therefore, it follows by Cramer's Theorem (cf. \cite[Theorem 2.2.3]{DemboZeitouni}) that
\begin{align}
\lim_{n\ra \infty} -\frac{1}{n}\log \Pr\left\{d(\bX,\bY) \leq D\right\} &= \bE_Z\left(\frac{D}{\sigma_X^2+\sigma_Y^2}\right).%
\end{align}
In this setting, it is a straightforward algebraic exercise to see that
\begin{align}
\bE_\ID(R,D,P_X,P_Y) < \bE_Z\left(\frac{D}{\sigma_X^2+\sigma_Y^2}\right)
\end{align}
for $R< \infty$ by putting
\begin{align}
&\rho_X =  \frac{\sigma_X^2 D + \sigma_Y^2(\sigma_X^2+\sigma_Y^2)}{(\sigma_X^2+\sigma_Y^2)^2}, &\rho_Y =  \frac{\sigma_Y^2 D + \sigma_X^2(\sigma_X^2+\sigma_Y^2)}{(\sigma_X^2+\sigma_Y^2)^2}
\end{align}
in \eqref{eqn:EIDdiffVar}.
Therefore, $\bE_\ID(R,D,P_X,P_Y)$ also precisely characterizes the best-possible exponent corresponding to the probability of a false positive event in this setting due to the relation \eqref{FPandMaybeRelation}.

In the case where $P_X=P_Y=N(\mu,\sigma^2)$, the symmetry in \eqref{eqn:EIDdiffVar} can be exploited to yield the following corollary.
\begin{cor}\label{cor:EID}
    Let  $P_X=P_Y=N(\mu,\sigma^2)$.  For any fixed rate $R>R_\ID(D,P_X,P_Y)$,
    \begin{align}
          \bE_\ID(R,D,P_X,P_Y) = \min_{\rho} ~2 \bE_Z(\rho) +\wp(R,D,{\rho \sigma^2},{\rho\sigma^2}),\label{EIDcor}
    \end{align}
    where the minimization is over all $\rho$ satisfying
\begin{align}
2\sigma^2\geq 2\rho\sigma^2 \geq D.
\end{align}
\end{cor}
A formal proof is given in Section \ref{sec:proofs}.   The identification exponent  \eqref{EIDcor} for the case of $D/\sigma^2 = 1.5$  is illustrated in Fig.~\ref{fig:EID}.
\begin{figure}
  \centering
  \includegraphics[width=6in]{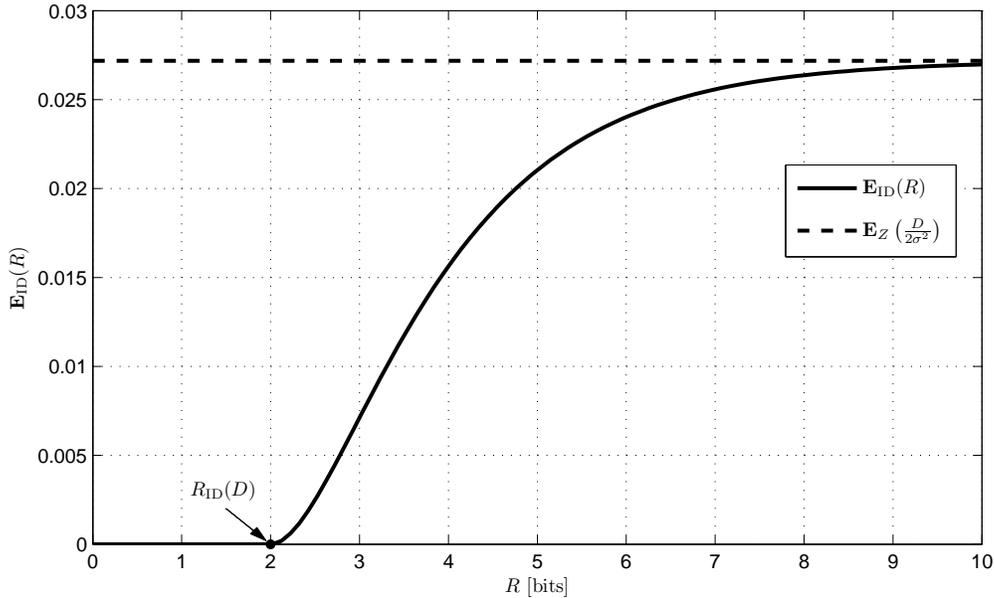}\\
  \caption{Plot of $\bE_\ID(R) :=  \bE_\ID(R,D,N(\mu,\sigma^2),N(\mu,\sigma^2))$ for $D/\sigma^2 = 1.5$.  In this case, $R_\ID(D,N(\mu,\sigma^2),N(\mu,\sigma^2)) = 2$ bits per symbol.}\label{fig:EID}
\end{figure}

Before proceeding, we briefly note that the identification exponent  $\bE_\ID(R,D,P_X,P_Y)$ can sometimes be strictly positive at $R\ra0$ \footnote{Note that whenever $R$ is \emph{equal} to zero, the probability of $\MAYBE$ is equal to $1$ (unless the supports of $P_X$ and $P_Y$ are disjoint in a way making \emph{any} two sequences $\bx$ and $\by$ dissimilar, making the problem degenerate).}. For instance, if
\begin{align}
\left|\frac{1}{\sqrt n}\EE\|\bX \| -\frac{1}{\sqrt n}\EE\|\bY \| \right| > \sqrt{D}+\epsilon
\end{align}
for some $\epsilon>0$, then the signature $T(\bX)$ can simply indicate whether $\left| \frac{1}{\sqrt n}\|\bX\|- \frac{1}{\sqrt n}\EE\|\bX \| \right| > \epsilon/2$, requiring rate $R=1/n$. Then, the query function $g$ returns $\MAYBE$ only if
\begin{align}
&\left| \frac{1}{\sqrt n}\|\bX\|- \frac{1}{\sqrt n}\EE\|\bX \| \right| > \epsilon/2, \mbox{~or} \label{exception1}\\
&\left| \frac{1}{\sqrt n}\|\bY\|- \frac{1}{\sqrt n}\EE\|\bY \| \right| > \epsilon/2.\label{exception2}
\end{align}
If neither \eqref{exception1} nor \eqref{exception2} occur, then it is readily verified that $d(\bX,\bY)>D$ using the triangle inequality.
Whenever the random variables $X^2$ and $Y^2$ satisfy a large deviations principle (as in the Gaussian case, and for many other distributions, cf. \cite{DemboZeitouni}), we see that $g$ returns $\MAYBE$ with probability exponentially decaying in $n$, and we can conclude that $\lim_{R\ra 0^+}\bE_\ID(R,D,P_X,P_Y) > 0$.  If this is indeed the case, then it also follows that $R_\ID(D,P_X,P_Y) = 0$ by definition.  Though this discussion applies for arbitrary distributions $P_X,P_Y$, this latter point is concretely reflected in Theorems  \ref{thm:RIDdiffVar} and \ref{thm:EIDdiffVar} for the case where $D\leq (\sigma_X-\sigma_Y)^2$.

\subsection{Upper Bounds on the Identification Rate}
In the previous two subsections, we focused our attention primarily to the case where $P_X$ and $P_Y$ were Gaussian distributions.  In the sequel, we consider more general distributions and show that Gaussian $P_X,P_Y$ constitute an extremal case in terms of the identification rate.

\begin{theorem}\label{thm:achGeneral}
Suppose $P_X$ and $P_Y$ are distributions with finite second moments $\sigma_X^2$ and $\sigma_Y^2$, respectively.  Then
\begin{align}
R_\ID(D,P_X,P_Y) \leq \overline{R}_\ID(D,P_X,P_Y) \triangleq \inf_{P_{\hat{X}|X}} I(X;\hat{X}),
\end{align}
where the infimum is taken over all conditional distributions $P_{\hat{X}|X}$ satisfying
\begin{align}
 \sqrt{\mathbb{E}\left[\left(\sqrt{\frac{\sigma_X}{\sigma_Y}}Y - \hat{X}\right)^2\right]}  \geq \sqrt{\mathbb{E}\left[\left(\sqrt{\frac{\sigma_Y}{\sigma_X}}X - \hat{X}\right)^2\right]}
+\sqrt{D-(\sigma_X-\sigma_Y)^2}\label{eqn:achIneq}
\end{align}
for $(Y,X,\hat{X}) \sim P_Y(y) P_X(x) P_{\hat{X}|X}(\hat{x}|x)$.  Moreover,
\begin{align}
\bE_\ID(R,D,P_X,P_Y)>0
\end{align}
for any $R >  \overline{R}_\ID(D,P_X,P_Y)$.
\end{theorem}

\begin{remark}
Note that Theorem \ref{thm:achGeneral} does not require $P_X$ and $P_Y$ to have identical means.
\end{remark}
\begin{remark}
Also note, that the achievability result and the proof technique carry over to general distortion criteria satisfying the triangle inequality. We omit the details as the focus of this paper is on the quadratic similarity criterion.
\end{remark}

For general source distributions $P_X, P_Y$, we lack a matching lower bound on $R_\ID(D,P_X,P_Y)$. However, such a converse was proved in the Gaussian setting (see Theorem \ref{thm:RIDdiffVar}). The key ingredient in the proof of Theorem \ref{thm:RIDdiffVar} is the isoperimetric inequality on the surface of a hypersphere -- the set on which the probability of a high dimensional Gaussian random vector concentrates (see Section \ref{sec:proofs} for details). In general,  precise isoperimetric inequalities are unknown and therefore establishing a general converse appears to be extremely difficult.

In spite of this, an application of Theorem \ref{thm:achGeneral} reveals the interesting fact that Gaussian $P_X$ and $P_Y$ correspond to sources which are ``most difficult" to compress for queries.  This is analogous to the setting of classical lossy compression, where the Gaussian source requires the maximum rate for compression subject to a quadratic distortion constraint. Formally,

\begin{theorem}\label{thm:GaussianExtreme}
Suppose $P_X$ and $P_Y$ have identical means and finite variances $\sigma_X^2$ and $\sigma_Y^2$, respectively.  Then
\begin{align}
R_\ID(D,P_X,P_Y) \leq R_\ID(D,N(0,\sigma_X^2),N(0,\sigma_Y^2)).
\end{align}
In particular, Gaussian $P_X$ and $P_Y$ demand the largest identification rate for given variances.
\end{theorem}

\subsection{Robust Identification Schemes}\label{subsec:RobustSchemes}
In addition to the extremal property of Gaussian sources described in Theorem \ref{thm:GaussianExtreme}, there exists a sequence of rate-$R$ identification schemes $\{T^{(n)},g^{(n)}\}_{n\ra \infty}$, where $(T^{(n)},g^{(n)})$ denotes a blocklength-$n$ identification scheme, designed for Gaussian sources which are \emph{robust} in the following sense.
Using the construction described in the achievability proof of Theorem \ref{thm:RIDdiffVar}, we can construct a sequence of $D$-admissible, rate-$R$ schemes $\{T^{(n)},g^{(n)}\}_{n\ra \infty}$ which satisfy
\begin{align}
      \lim_{n\ra\infty} \Pr\left\{g^{(n)}\left(T^{(n)}(\bX),\bY \right) = \MAYBE\right\} = 0
\end{align}
when $\bX,\bY \sim \prod_{i=1}^nP_X(x_i) P_Y(y_i)$,  $P_X = N(0,\sigma_X^2)$, $P_Y = N(0,\sigma_Y^2)$ and
\begin{align}
R > R_\ID\left(D,N(0,\sigma_X^2),N(0,\sigma_Y^2)\right).
\end{align}

  It turns out that this particular sequence $\{T^{(n)},g^{(n)}\}_{n\ra \infty}$ is robust to the source distributions in the sense that we also have
\begin{align}
      \lim_{n\ra\infty} \Pr\left\{g^{(n)}\left(T^{(n)}(\tilde\bX),\tilde\bY \right) = \MAYBE\right\} = 0
\end{align}
when $\tilde\bX,\tilde\bY \sim \prod_{i=1}^nP_{\tilde X}(\tilde x_i) P_{\tilde Y}(\tilde y_i)$, and $P_{\tilde X}$, $P_{ \tilde Y}$ are zero-mean distributions with variances $\sigma_X^2$ and $\sigma_Y^2$, respectively.  Moreover, the sequence $\{T^{(n)},g^{(n)}\}_{n\ra \infty}$ continues to be $D$-admissible for the sources $\tilde\bX,\tilde\bY$.  Thus, roughly speaking, a scheme $(T,g)$ which is ``good" for Gaussian sources $\bX,\bY$ can be expected to perform well for arbitrary sources $\tilde \bX,\tilde \bY$, provided the respective variances match their Gaussian counterparts and the blocklength $n$ is large.  The proof of this robustness property is given in Section \ref{subsec:RobustProof}.

\section{Proofs}\label{sec:proofs}
In this section, we prove each of the main results.  Proofs are organized by subsection.  We begin with a primer on the key geometric ideas that are used throughout the proofs.
\subsection{Geometric Preliminaries}
For the proofs we require the following definitions related to $n$-dimensional Euclidean geometry.

For $r>0,\bu\in\Reals^n$, let $\BALL_r(\bu)\subseteq \Reals^n$ denote the ball with radius $r$ centered at $\bu$:
\begin{equation}
    \BALL_r(\bu) \triangleq \left\{\bx\in\Reals^n : \|\bx-\bu\| \leq  r\right\}.
\end{equation}
$\BALL_r(\mathbf{0})$ will be denoted $\BALL_r$.

Denote by $S_r\subseteq \Reals^n$ the spherical shell with radius $r$ centered at the origin:
\begin{equation}
    S_r \triangleq \left\{\bx\in\Reals^n : \|\bx\| = r\right\}.
\end{equation}

For any two vectors $\bx_1,\bx_2 \in \Reals^n \setminus \{\mathbf{0}\}$, the angle between them shall be denoted by
\begin{equation}\label{eqn:defAngle}
  \angle(\bx_1,\bx_2) \triangleq
  \arccos\left(\frac{\bx_1^T\bx_2}{\|\bx_1\|\|\bx_2\|}\right) \in [0,\pi].
\end{equation}

For $\theta \in [0,\pi]$ and a point $\bu \in \Reals^n\setminus \{\mathbf{0}\}$, define the cone with half angle $\theta$ and axis going through $\bu$:
\begin{equation}
  \CONE(\bu,\theta) \triangleq \left\{\bx\in\Reals^n :  \angle(\bu,\bx) \leq \theta \right\}.
\end{equation}
Note that $\CONE(\bu,0)$ is the half-infinite line $\{\alpha \bu : \alpha > 0\}$, that
$\CONE(\bu,\pi/2)$ is the half-space containing $\bu$ that is bordered by the  hyperplane  orthogonal to $\bu$ which passes through the origin,
and that $\CONE(\bu,\pi)$ is the entire space $\Reals^n$.
Also, note that $\CONE(\bu_1,\theta) = \CONE(\bu_2,\theta)$ for any $\bu_1=\lambda\bu_2$, $\lambda >0$.

For $r>0, \bu\in\Reals^n\setminus \{\mathbf{0}\}$ and $\theta \in [0,\pi]$, denote by $\CAP_r(\bu,\theta)$ the spherical cap:
\begin{equation}
  \CAP_r(\bu,\theta) \triangleq S_r \cap \CONE(\bu,\theta).
\end{equation}

Let $\Omega(\theta)$ denote the fraction of the (hyper-)surface area of $S_r$ that is occupied by $\CAP_r(\bu,\theta)$:
\begin{equation}
  \Omega(\theta) \triangleq \frac{|\CAP_r(\bu,\theta)|}{|S_r|}.
\end{equation}
Note that the value of $\Omega(\theta)$ depends neither on $r$ nor on $\bu$.  The following bounds on  $\Omega(\theta)$ will be useful:

\begin{lem}{\cite[Corrolary 3.2]{boroczky03}}\label{lem:Omega}
    For $0 < \theta < \arccos(1/\sqrt{n}) < \frac{\pi}{2}$, we have
\begin{align}
  \Omega(\theta) &< \frac{1}{\sqrt{2\pi(n-1)}}\cdot\frac{1}{\cos \theta} \cdot\sin^{n-1}\theta, \label{eqn:OmegaUpper}\\
  \Omega(\theta) &>  \frac{1}{3\sqrt{2\pi n}}\cdot\frac{1}{\cos \theta} \cdot\sin^{n-1}\theta. \label{eqn:OmegaLower}
\end{align}
\end{lem}

\bigskip

For positive $r_1 \leq r_2 \in \Reals$, let $S_{r_1,r_2}\subseteq \Reals^n$ be a
spherical shell of inner radius $r_1$ and outer radius $r_2$:
\begin{equation}
    S_{r_1,r_2} \triangleq \left\{\bx\in\Reals^n : r_1 \leq \|\bx\| \leq r_2\right\}.
\end{equation}
For a given half-angle $\theta \in [0,\pi]$, define the $(r_1,r_2)$-spherical cap with half-angle $\theta$ and axis going through $\bu$ as
\begin{align}
  \CAP_{r_1,r_2}(\bu,\theta)
  &\triangleq \CONE(\bu,\theta) \cap S_{r_1,r_2}.
\end{align}

For a set $A \subseteq \Reals^n$ and $D>0$, the $D$-expansion of $A$, denoted $\Gamma^D(A)$ is defined as
\begin{align}
  \Gamma^D(A)
  &\triangleq \{\by\in\Reals^n : \exists_{\bx\in A}  d(\bx,\by)\leq D\}\\
  &= A + \BALL_{\sqrt{nD}},
\end{align}
where we have used $+$ to denote the Minkowski sum.

\subsection{Codes that cover a spherical shell}

\begin{defn}
Let $S_{r}\subseteq \Reals^n$ be the spherical shell with radius $r$.  We say that a set of points $\cC=\{\bu_1, \dots, \bu_m : \bu_i \in \Reals^n\}$  is a code that $D$-covers $S_{r}$ if
\begin{align}
S_{r} \subseteq \bigcup_{\bu\in \cC} \BALL_{\sqrt{nD}}(\bu).
\end{align}
The \emph{rate} of $\mathcal{C}$ is defined as $\frac{1}{n}\log m$.
\end{defn}

When not explicitly stated, the ambient dimension $n$ of the code $\cC$ will be clear from context.

\begin{lem}[Following \cite{Dumer07}] \label{lem:ShellCovering}
Fix $\sigma^2>0$ and the dimension $n$.   For any $0<D_0<\sigma^2$, there exists a code $\cC$ that $D_0$-covers $S_{\sqrt{n\sigma^2}}$ with rate
\begin{equation}\label{eqn:ShellCoveringExistence}
  R_0 = \frac{1}{n}\log |\cC| \leq \frac{1}{2}\log\frac{\sigma^2}{D_0} + O\left(\frac{\log n}{n} \right).
\end{equation}
Moreover, for all $\bu\in\cC$, we have $\|\bu\| = \sqrt{n(\sigma^2-D_0)}$, and
\begin{align}
\CAP_{\sqrt{n\sigma^2}}(\bu,\theta_0) = S_{\sqrt{n\sigma^2}} \cap  \BALL_{\sqrt{nD_0}}(\bu),
\end{align}
where
\begin{equation}\label{eqn:def_theta0}
  \theta_0 \triangleq \arcsin(\sqrt{D_0/\sigma^2})<\frac{\pi}{2}.
\end{equation}
\end{lem}
\begin{proof} Appendix~\ref{app:ShellCovering}.\end{proof}

It is no surprise that the term  $\frac{1}{2}\log\frac{\sigma^2}{D_0}$ appearing in \eqref{eqn:ShellCoveringExistence} is identical to the rate-distortion function for the Gaussian source with variance $\sigma^2$ evaluated at distortion-level $D_0$.
We could have therefore used any standard (random code-like) construction.
However, using Lemma \ref{lem:ShellCovering} will be more convenient for our purposes since each point in $S_r$ is guaranteed to be covered, and hence we do not need to account for another error event. This fact will make the subsequent proofs more straightforward.

\subsection{Identification Rate}
The proof of Theorem \ref{thm:RIDdiffVar} is somewhat lengthy, so we first give the key ideas here before moving onto the formal details.

The proof of the theorem relies on the fact that a high-dimensional Gaussian random vector -- with independent entries having zero mean and variance $\sigma_X^2$ -- concentrates near a thin hyper-spherical shell of radius $r_0 \triangleq \sqrt{n\sigma_X^2}$, which we call the \emph{typical sphere}.  The signature assignment constructed in the direct part of the proof quantizes the surface of the typical sphere into regions roughly described by spherical caps.  The query function $g$, knowing which cap $\bX$ lies in from the received signature, returns $\MAYBE$ only if $\bY$ lies within Euclidean distance $\sqrt D$ of the cap in which $\bX$ lies.  Thus, the goal in the direct part is to show that, for sufficiently large rate $R$, the probability $\bY$ falls into the $\Gamma^D$-expansion of any given cap is vanishing.

The key ingredient in proving the converse is the isoperimetric inequality on the surface of the hypersphere, known as Levy's lemma (see e.g. \cite[Theorem 1.1]{ledoux2011probability}).  In a nutshell, we apply Levy's lemma to prove that any given identification system $(T,g)$ requires a rate that is essentially as large as an identification system that uniquely assigns caps on the typical sphere to signatures (as is done by the achievability scheme).  The apparent need for a refined isoperimetric inequality to prove the converse distinguishes our problem from the class of standard rate-distortion problems.

\begin{proof}[Proof of Theorem~\ref{thm:RIDdiffVar}]
Before beginning the proof, we first note that it is sufficient to consider $D$ in the interval $(\sigma_X-\sigma_Y)^2 < D < \sigma_X^2+\sigma_Y^2$.  The claims that $R_\ID(D,P_X,P_Y) = 0$ for $D \leq (\sigma_X-\sigma_Y)^2$, and $R_\ID(D,P_X,P_Y) = \infty$  for $D \geq \sigma_X^2+\sigma_Y^2$ then follow from monotonicity of $R_\ID(D,P_X,P_Y)$ in $D$.

\textbf{Direct Part:}  Fix a small $\epsilon>0$, and define $r_X \triangleq \sqrt{n\sigma_X^2}$ (i.e., the radius of the typical sphere).  Let $D$ be a desired similarity threshold in the interval $(\sigma_X-\sigma_Y)^2 < D < \sigma_X^2+\sigma_Y^2$, and let $\eta>0$ be sufficiently small so that
\begin{align}
(1-\epsilon) \left[\frac{\sigma_X^2+\sigma_Y^2-D}{2\sigma_X \sigma_Y}\right]^2 <  \left[\frac{\sigma_X^2+\sigma_Y^2-2\eta-D}{2\sqrt{(\sigma_X^2+\eta)(\sigma_Y^2+\eta)}}\right]^2. \label{etaSuffSmall}
\end{align}
Next, define a constant $D_0$ satisfying
\begin{align}
(1-\epsilon) \sigma_X^2 \left[\frac{\sigma_X^2+\sigma_Y^2-2\eta-D}{2\sqrt{(\sigma_X^2+\eta)(\sigma_Y^2+\eta)}}\right]^2  < D_0 < \sigma_X^2 \left[\frac{\sigma_X^2+\sigma_Y^2-2\eta-D}{2\sqrt{(\sigma_X^2+\eta)(\sigma_Y^2+\eta)}}\right]^2. \label{eqn:DefnD0}
\end{align}
The motivation behind the choices of $\eta$ and $D_0$ satisfying \eqref{etaSuffSmall} and \eqref{eqn:DefnD0} will become clear as the proof proceeds.

By our assumption that $D> (\sigma_X-\sigma_Y)^2$, it follows that $0 < D_0<\sigma_X^2$.  By Lemma~\ref{lem:ShellCovering}, there exists  a code $\cC$ which $D_0$-covers $S_{r_X}$ with rate $R_0$ bounded by
\begin{align}
  R_0 \leq \frac{1}{2}\log\frac{\sigma_X^2}{D_0} + O\left(\frac{\log n}{n} \right).
\end{align}

Let $T_0:S_{r_X}\ra\cC$ be the quantization operation defined by
\begin{align}
T_0(\bx) = \arg \min_{\bu\in \cC}\|\bx-\bu\| \mbox{~~for $\bx\in S_{r_X}$.}
\end{align}
That is, the function $T_0(\bx)$ maps $\bx\in S_{r_X}$ to the closest reconstruction point $\bu \in \cC$. Since $\cC$ is a code that $D_0$-covers $S_{r_X}$, it follows that
\begin{align}
 \|T_0(\bx)-\bx\| \leq \rho_0 \triangleq \sqrt{nD_0} \mbox{~~for all $\bx\in S_{r_X}$.}
\end{align}
Denote the points in $S_{r_X}$ that are mapped to $\bu$ by $T_0^{-1}(\bu)$. With this notation, it follows by construction that
\begin{equation}
  T_0^{-1}(\bu) \subseteq \CAP_{r_X}(\bu,\theta_0),
\end{equation}
where $\theta_0 \triangleq \arcsin(\sqrt{D_0/\sigma_X^2})$ courtesy of Lemma \ref{lem:ShellCovering}. The set $\CAP_{r_X}(\bu,\theta_0)$ is illustrated in Fig.~\ref{fig:circle_r0}.
\begin{figure}
  \centering
  \def\svgwidth{5in}
  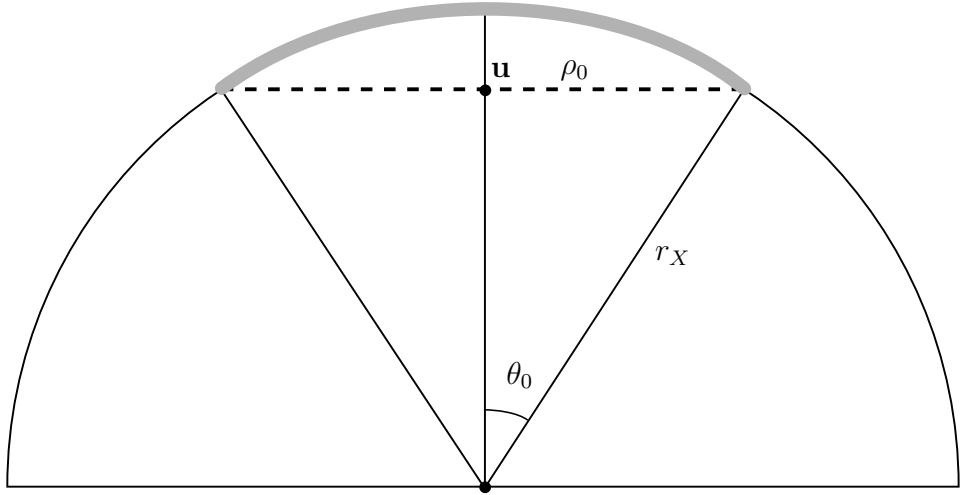
  \caption{Illustration of a single cap $\CAP_{r_X}(\bu,\theta_0)$ (denoted in grey).}\label{fig:circle_r0}
\end{figure}

\bigskip
Define $\ST_X$ to be the set of all vectors $\bx\in\Reals^n$ s.t. $\sigma_X^2-\eta \leq \frac{1}{n}\|\bx\|^2\leq \sigma_X^2+\eta$. In other words,
\begin{equation}
  \ST_X \triangleq S_{r^-,r^+}, \label{sTypXdefn}
\end{equation}
where  $r^\pm \triangleq \sqrt{n(\sigma_X^2\pm\eta)}$.
Note that $\Pr\{\bX \notin \ST_X\}$ vanishes with $n$ (in fact, it vanishes exponentially), which motivates the notation $\ST_X$.

Next, we construct a mapping $T:\ST_X \ra \cC$ defined as follows:
\begin{equation}
    T(\bx) = T_0\left(\bx \cdot \frac{\sqrt{n\sigma_X^2}}{\|\bx\|}\right).
\end{equation}
Since $T_0^{-1}(\bu)$ is contained in $\CAP_{r_X}(\bu,\theta_0)$, we similarly have that the inverse map $T^{-1}$ satisfies
\begin{equation}\label{eqn:cap_contained}
  T^{-1}(\bu) \subseteq \CAP_{r^-,r^+}(\bu,\theta_0).
\end{equation}

The signature assignment for our identification scheme for $\bx\in\ST_X$ shall be given by the function $T(\cdot)$ defined above.  For $\bx \notin \ST_X$ we define $T(\bx) = \ERASURE$, where $\ERASURE$ is an additional ``erasure" symbol, denoting the fact that the signature does not convey any information about $\bx$ in this case (and the decision function $g(\cdot,\cdot)$ must output $\MAYBE$).
Note that the additional rate incurred by the erasure symbol is negligible and we still have that the signature assignment's rate $R$ is bounded by
\begin{align}
R &= \frac{1}{n}\log\left( |\cC| + 1 \right) \\
&\leq \frac{1}{2}\log\frac{\sigma_X^2}{D_0} + O\left(\frac{\log n}{n}\right)\\
&\leq \log\frac{2\sigma_X\sigma_Y}{\sigma_X^2+\sigma_Y^2-D} + \log\frac{1}{1-\epsilon} + O\left(\frac{\log n}{n}\right), \label{upperBoundRThm1}
\end{align}
where the final inequality follows from \eqref{etaSuffSmall} and  \eqref{eqn:DefnD0}.

The query function $g(\cdot,\cdot)$ is defined to be the optimal one given the signature mapping $T(\cdot)$:
\begin{equation}\label{eqn:opt_decoder}
  g(t,\by) = \left\{
                 \begin{array}{ll}
                   \MAYBE & \hbox{If $t = \ERASURE$ or if $\exists \bx' \in T^{-1}(t) \mbox{ s.t. } d(\bx',\by)\leq D$}
                    \\
                   \NO & \hbox{otherwise.}
                 \end{array}
               \right.
\end{equation}

\bigskip
Using the shorthand notation
\begin{align}
  \Pr\{\MAYBE\} \triangleq \Pr \{ g(T(\bX),\bY) = \MAYBE \},
\end{align}
we analyze $\Pr\{\MAYBE\}$  as follows. First, define a typical set for the $\bY$-sequences:
\begin{equation}
  \ST_Y \triangleq S_{r_Y^-,r_Y^+},
\end{equation}
where  $r_Y^\pm \triangleq \sqrt{n(\sigma_Y^2\pm\eta)}$, and write
\begin{align}
  \Pr\{\MAYBE\} &
  \leq \Pr\{\MAYBE | \bX \in \ST_X, \bY \in \ST_Y\}+ \Pr\{\bX \notin \ST_X\} + \Pr\{\bY \notin \ST_Y\}. \label{eqn:totalMaypeProbAchievability}
\end{align}

Note that the latter two terms in \eqref{eqn:totalMaypeProbAchievability} vanish as $n$ grows large, thus we focus on bounding the first term.  To this end, we require the following lemma.

\begin{lem}\label{lem:Dexpansion}
Let $\cC$ and $\eta$ be as defined above.  For any $\bu \in \cC$, we have
\begin{equation}
    \Gamma^D\left(T^{-1}(\bu)\right) \cap \ST_Y\subseteq \CONE(\bu,\theta'),
\end{equation}
where
\begin{equation}\label{eqn:thetaprime}
  \theta' \triangleq  \theta_0+\theta_1 < \frac{\pi}{2},
\end{equation}
and the angles $\theta_0$ and $\theta_1$ are given by
\begin{align}
  \theta_0 &\triangleq \arcsin\left(\sqrt{\frac{D_0}{\sigma_X^2}}\right)\label{eqn:def_ThetaZero}\\
  \theta_1 &\triangleq \arccos\left(\frac{\sigma_X^2+\sigma_Y^2-2\eta-D}{2\sqrt{(\sigma_X^2+\eta)(\sigma_Y^2+\eta)}}\right).\label{eqn:def_theta1}
\end{align}
\end{lem}
\begin{proof} Appendix~\ref{app:Dexpansion}.
\end{proof}
Fig.~\ref{fig:LawOfCosines_rXrY} illustrates the claim in the lemma.
\begin{figure}
  \centering
  \def\svgwidth{5in}
  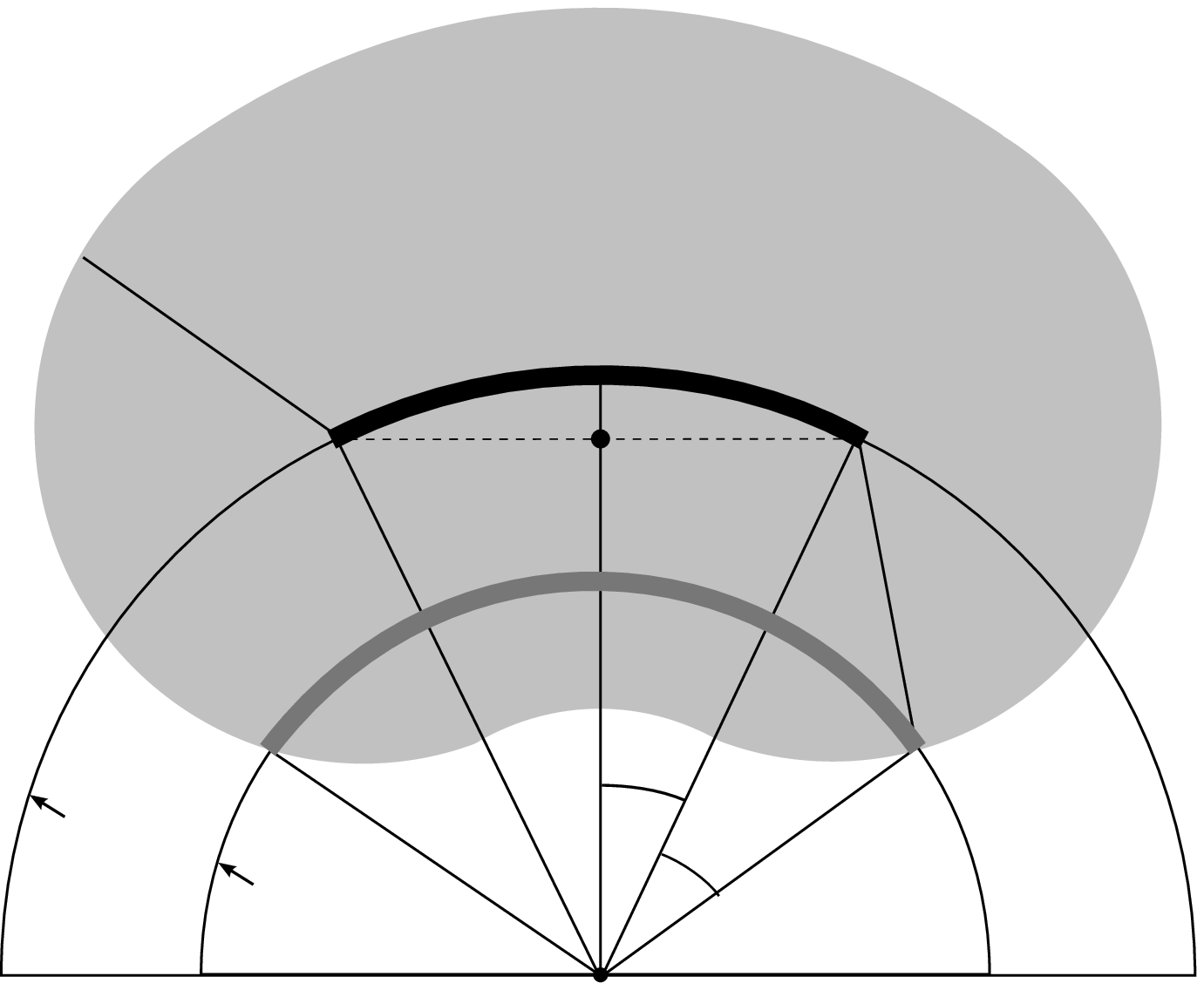
  \caption{Illustration for Lemma~\ref{lem:Dexpansion}. The black region marks $\CAP_{r^-,r^+}(\bu,\theta_0)$. The grey area denotes $\Gamma^D\left(\CAP_{r^-,r^+}(\bu,\theta_0)\right)$, and the dark grey region marks the intersection $\Gamma^D\left(\CAP_{r^-,r^+}(\bu,\theta_0)\right) \cap \ST_Y$.
  }\label{fig:LawOfCosines_rXrY}
\end{figure}

Let $\theta'$ be as defined in Lemma \ref{lem:Dexpansion} above. We continue with
\begin{align}
  \!\!\Pr\{\MAYBE | \bX \in \ST_X, \bY \in \ST_Y\}
  &\overset{(a)}= \Pr\left\{\bY \in  \Gamma^D\left(T^{-1}(T(\bX))\right) | \bX\in \ST_X, \bY \in \ST_Y \right\} \notag\\
  &\overset{(b)}\leq  \Pr\left\{\bY \in  \CONE(T(\bX),\theta')|  \bX\in \ST_X, \bY \in \ST_Y\right\} \notag\\
  &\overset{(c)}= \Omega(\theta') \notag\\
  &\overset{(d)}\leq \frac{1}{\sqrt{2\pi (n-1)}} \cdot \frac{1}{\cos\theta'} \cdot \sin^{n-1}\theta'. \label{punchlineInequalityThm1Direct}
\end{align}
Identity $(a)$ follows by definition of the query function $g(\cdot,\cdot)$.  Inequality $(b)$ follows from Lemma~\ref{lem:Dexpansion}. Equality $(c)$ follows since $\bY$ is uniformly distributed within each shell $S_r$ of radius $r>0$ (due to the spherical symmetry of the Gaussian distribution), and the probability of falling in a cap of a given half-angle $\theta'$ is precisely the fraction of the surface that is occupied by the cap, $\Omega(\theta')$. Inequality $(d)$ follows since $\theta' \leq \arccos(1/\sqrt n)$ for sufficiently large $n$, and therefore  \eqref{eqn:OmegaUpper} applies.

Since $\theta'<\pi/2$, we have $\sin \theta' < 1$, and it therefore follows from \eqref{punchlineInequalityThm1Direct} that the probability $\Pr\{\MAYBE | \bX \in \ST, \bY \in \ST_Y\}$ vanishes with $n$. Thus, since $\epsilon$ was arbitrary, recalling \eqref{upperBoundRThm1} completes the direct part of the proof.
\begin{remark}
The alert reader will observe that the direct part also follows from the direct part of Theorem~\ref{thm:EIDdiffVar}. However, we have chosen to include an explicit proof here to introduce the notations and ideas crucial for proving Theorem~\ref{thm:EIDdiffVar}.
\end{remark}

\bigskip
\textbf{Converse Part:}  Let $\eta>0$ and define $\ST_X$ as in \eqref{sTypXdefn}.
Let $T:\Reals^n\ra\{1,...,2^{nR}\}$ be a given signature function corresponding to a $D$-admissible system $(T,g)$, and assume that
\begin{align}
 \Pr \{ g(T(\bX),\bY) = \MAYBE \} \leq \frac{1}{4} \label{UBpmaybe_fourth}
\end{align}
since we are only interested in $D$-achievable rates $R$.  As before, we will use the shorthand notation $ \Pr\{\MAYBE\} \triangleq \Pr \{ g(T(\bX),\bY) = \MAYBE \}$ to simplify the presentation.

We shall restrict our attention to the typical sphere. To this end, define the mapping $\tilde T:\ST_X \ra \{1,...,2^{nR}\}$, where $\tilde{T}(\bx)=T(\bx)$ for $\bx\in \ST_X$. Let $\tilde T^{-1}(\cdot)$ denote the inverse mapping of $\tilde T(\cdot)$, i.e.
\begin{align}
  \tilde T^{-1}(i)
  &\triangleq \{\bx \in \ST_X : T(\bx) = i\} \\
  &= T^{-1}(i) \cap \ST_X.
\end{align}

Let $p_i \triangleq \Pr\{\bX \in \tilde T^{-1}(i) | \bX \in \ST_X\}$. Clearly, we have $\sum_{i=1}^{2^{nR}}p_i = 1$. Define the set $A_i \subseteq S_{r_X}$ to be projection of $\tilde T^{-1}(i)$ onto the sphere $S_{r_X}$:
\begin{equation}
  A_i = \left\{ r_X \frac{\bx}{\|\bx\|} : \bx \in \tilde T^{-1}(i) \right\}
\end{equation}

Let $\alpha_i$ denote the fraction of the surface area of $S_{r_X}$ that is occupied by $A_i$. By the spherical symmetry of the pdf of $\bX$, $\alpha_i$ is also equal to the probability that the projection of $\bX$ onto $S_{r_0}$ lies in $A_i$. Therefore $\alpha_i \geq p_i$, with equality if and only if $\tilde T^{-1}(i)$ is a thick cap with inner and outer radii $r^\pm \triangleq \sqrt{n(\sigma_X^2\pm\eta)}$.

Let $D' \triangleq (\sqrt{D} +\sqrt{\sigma_X^2-\eta} -  \sqrt{\sigma_X^2})^2 < D$. It can easily be verified that
\begin{equation}\label{eqn:AiProperty}
  \Gamma^{D'}\left(A_i\right) \subseteq \Gamma^D\left(\tilde T^{-1}(i)\right).
\end{equation}

Now let $D'' \triangleq (\sqrt{D'} + \sqrt{\sigma_Y^2-\eta}-\sqrt{\sigma_Y^2})^2$, and let the set $B_i$ denote the $D''$-expansion of $A_i$, restricted to the sphere $S_{r_Y}$, i.e.
\begin{equation}
  B_i \triangleq \Gamma^{D''}\left(A_i\right) \cap S_{r_Y}.
\end{equation}
The set $B_i$ can also be thought of an expansion of a set $\tilde A_i \triangleq \frac{\sigma_Y}{\sigma_X} \cdot A_i$,
with the alternative distance measure $\breve{d}(\cdot,\cdot)$ defined over the sphere $S_{r_Y}$ that measures the arc-length between the two points (i.e., the geodesic distance).
Also note that $\alpha_i = \frac{|A_i|}{|S_{r_X}|} = \frac{|\tilde A_i|}{|S_{r_Y}|}$ where $|\cdot|$ is used to denote the (hyper-) surface area.
Let $\beta_i = \frac{|B_i|}{|S_{r_0}|}$ denote the fraction of $S_{r_0}$ that is occupied by $B_i$.

Let the set $C_i$ denote the $r_Y^-,r_Y^+$ thickening of $B_i$ as follows:
\begin{equation}
  C_i = \left\{\by\in \ST_Y : r_Y\frac{\by}{\|\by\|} \in B_i\right\}.
\end{equation}
Next, it can also be verified that
\begin{equation}\label{eqn:CiProperty}
  C_i \subseteq \Gamma^{D'}\left(A_i\right).
\end{equation}
\bigskip

Suppose that $\bx \in \ST_X$ and that $T(\bx) = i$. Then we have:

\begin{align*}
  \Pr\{\MAYBE | \bX =\bx \in \ST_X \}
  &\geq \Pr\left\{\bY \in  \Gamma^D\left(T^{-1}(i)\right)\right\} \\
  &\overset{(a)}\geq \Pr\left\{\bY \in  \Gamma^D\left(\tilde T^{-1}(i)\right)\right\} \\
  &\overset{(b)}\geq \Pr\left\{\bY \in  \Gamma^{D'}\left(A_i\right)\right\} \\
  &\overset{(c)}\geq \Pr\left\{\bY \in C_i \right\},
\end{align*}
where $(a)$ follows since $\tilde T^{-1}(i) \subseteq T^{-1}(i)$, and $(b)$ and $(c)$ follow from \eqref{eqn:AiProperty} and \eqref{eqn:CiProperty} respectively.

Let $f_{\bY}$ be the density of $\bY$.  Then, we continue with
\begin{align*}
  \Pr\left\{\bY \in C_i \right\}
  &= \int_{C_i} f_\bY(\by)d\by = \beta_i \cdot \Pr\{\bY \in \ST_Y\},
\end{align*}
where the second equality follows from the spherical symmetry of $f_\bY(\by)$.

\medskip

We now arrive at the main step in proving the converse. The key ingredient we require is the well-known isoperimetric inequality on the hypersphere (cf. \cite[Theorem 1.1]{ledoux2011probability}) which states that, among all subsets of the hypersphere with a given surface area, spherical caps have minimum $D$-expansion measured under geodesic distance.  As noted before, the set $B_i \subseteq S_{r_Y}$ is an expansion of the set $\tilde A_i \subseteq S_{r_Y}$ with the arclength (i.e.,  geodesic) distance measure.   Therefore, it follows from the  isoperimetric inequality that
\begin{align}
 | B_i | = \left| \Gamma^{D''}\left(A_i\right) \cap S_{r_Y} \right| &\geq \left| \Gamma^{D''}\left(\CAP_{r_X}(\bu,\theta_i) \right) \cap S_{r_Y} \right|\\
 &=\left| \CAP_{r_Y}(\bu,\theta_i + \theta_{D''}) \right|,
\end{align}
where $\bu$ is an arbitrary point and
\begin{align}
 \theta_i &\triangleq  \Omega^{-1}(\alpha_i)\\
 \theta_{D''} &\triangleq \arccos\left(\frac{\sigma_X^2 + \sigma_Y^2-D''}{2\sigma_X\sigma_Y}\right).
\end{align}
Therefore, we can conclude that if $\bx \in \ST_X$ and $T(\bx) = i$, then
\begin{equation}
    \Pr\{\MAYBE | \bX =\bx\} \geq \Pr\{\bY \in \ST_Y\}\cdot
    \Omega\left(\theta_{D''}+\Omega^{-1}(\alpha_i)\right).
\end{equation}

Now, the average quantity $\Pr\{\MAYBE | \bX \in \ST_X \}$ is bounded as follows
\begin{align}
  \Pr\{\MAYBE | \bX \in \ST_X \}
  &= \sum_{i=1}^{2^{nR}} \Pr\{T(\bX) = i | \bX \in \ST_X\} \Pr\{\MAYBE | T(\bX) = i , \bX \in \ST_X\} \notag\\
  &\geq \sum_{i=1}^{2^{nR}} p_i \cdot \Pr\{\bY \in \ST_Y\}\cdot
    \Omega\left(\theta_{D''}+\Omega^{-1}(\alpha_i)\right)\\
  &\geq \Pr\{\bY \in \ST_Y\}\cdot\sum_{i=1}^{2^{nR}} p_i \cdot
    \Omega\left(\theta_{D''}+\Omega^{-1}(p_i)\right),\label{eqn:sum}
\end{align}
where the last inequality follows since $\alpha_i \geq p_i$ and  the function $\Omega(\theta_{D''} + \Omega^{-1}(\cdot))$ is monotone increasing.

If the scheme at hand were to satisfy $p_i = 2^{-nR}$ for all $i$, then we could simply continue with analyzing $\Omega\left(\theta_{D''}+\Omega^{-1}(2^{-nR})\right)$.
However, in general this might not be the case. We therefore require the following lemma:
\begin{lem}\label{lem:Markov}
    Let $0<\Omega^*<1$ and $0 < c < 1$ be given constants. Define $p^*$ to be the solution to $\Omega(\theta_{D''} + \Omega^{-1}(p))=\Omega^*$. Then if
    \begin{equation}\label{eqn:Sum_pi_Omega}
      \sum_{i=1}^{2^{nR}} p_i \cdot
    \Omega\left(\theta_{D''}+\Omega^{-1}(p_i)\right) \leq c\cdot \Omega^*,
    \end{equation}
    then
    \begin{equation}\label{eqn:Markov}
      R \geq \frac{1}{n}\log \frac{1-c}{p^*}.
    \end{equation}
\end{lem}
\begin{proof} Appendix~\ref{app:Markov}.
\end{proof}

For our purposes\footnote{We shall use Lemma~\ref{lem:Markov} again for proving the identification exponent results, but with a different $\Omega^*$.}
we set $\Omega^*=\frac{1}{2}$ so that $\Omega(\theta_{D''} + \Omega^{-1}(p^*))=\frac{1}{2}.$ Now use \eqref{eqn:OmegaUpper} to upper bound $\Omega(\cdot)$ and evaluate $p^*$:

\begin{align*}
  p^*
  &= \Omega\left(\frac{\pi}{2} - \theta_{D''}\right)\\
  &\leq \frac{1}{\sqrt{2\pi(n-1)}}\cdot\frac{1}{\cos\left(\frac{\pi}{2} - \theta_{D''}\right)} \cdot\sin^{n-1}\left(\frac{\pi}{2} - \theta_{D''}\right)\\
  &\leq \frac{1}{\sqrt{2\pi(n-1)}} \cdot\cos^{n-1}\left(\theta_{D''}\right).
\end{align*}
Recalling the definition of $\theta_{D''}$, we have
\begin{align*}
  \cos\left(\theta_{D''}\right) &=\frac{\sigma_X^2 + \sigma_Y^2-D''}{2\sigma_X\sigma_Y},
\end{align*}
therefore
\begin{equation}\label{eqn:Rofpstar}
  \frac{1}{n}\log\frac{1}{p^*} = \log\frac{2\sigma_X\sigma_Y}{\sigma_X^2 + \sigma_Y^2-D''} + O\left(\frac{\log n}{n}\right).
\end{equation}

Our goal, now, is to show that the rate $R$ must be \emph{lower} bounded by the identification rate from \eqref{eqn:RIDdiffVar}.  Recalling \eqref{UBpmaybe_fourth}, it follows that
\begin{align}
  \frac{1}{4}
  &\geq \Pr\{\MAYBE\} \\
  &= \Pr\{\bX \in \ST_X\} \cdot \Pr\{\MAYBE | \bX \in \ST_X\} \\
  &\quad\quad+\Pr\{\bX \notin \ST_X\}\cdot \Pr\{\MAYBE | \bX \notin \ST_X\}\\
  &\geq \Pr\{\bX \in \ST_X\} \cdot \Pr\{\MAYBE | \bX \in \ST_X\}\\
  &\geq \Pr\{\bX \in \ST_X\} \cdot \Pr\{\bY \in \ST_Y\}\cdot\sum_{i=1}^{2^{nR}} p_i \cdot
    \Omega\left(\theta_{D''}+\Omega^{-1}(p_i)\right),
\end{align}
where the final inequality is simply \eqref{eqn:sum}.

Since $\Pr\{\bY \in \ST_Y\}$ and $\Pr\{\bX \in \ST_X\}$ both approach $1$ as $n$ grows, we may assume that both probabilities are above $\frac{3}{4}$ (for large enough $n$).
Then, we can now invoke Lemma~\ref{lem:Markov} with  $c=8/9$ and $\Omega^*=1/2$, combined with \eqref{eqn:Rofpstar}, to conclude that
\begin{equation}
  R \geq \log\frac{2\sigma_X\sigma_Y}{\sigma_X^2 + \sigma_Y^2-D''}
  + O\left(\frac{\log n}{n}\right).
\end{equation}
As $\eta$ can be taken to be arbitrarily small, $D''$ can be arbitrarily close to $D$, completing the proof of the converse.
\end{proof}

\subsection{Identification Exponent}

As with Theorem \ref{thm:RIDdiffVar}, the proof of Theorem \ref{thm:EIDdiffVar} is rather involved, so we first sketch the main ideas before moving on to the formal proof. Characterizing the optimal exponent requires a slightly more sophisticated scheme than characterizing the identification rate, but the proofs are very similar in spirit.

The achievability proof builds upon that of Theorem \ref{thm:RIDdiffVar} in the sense that we refine the signature assignment to quantize $\bx/\|\bx\|$ and $\|\bx\|$ separately.  Intuitively, we can think of our scheme as quantizing the \emph{direction} and \emph{amplitude} of the vector $\bx$ (similarly to `shape-gain' quantizers \cite[Ch. 12]{GershoGray}).  Similar to the achievability proof of Theorem \ref{thm:RIDdiffVar}, the set of vectors $\bx/\|\bx\|$ are quantized by covering the unit sphere with regions roughly described by caps. It will turn out that the achievable identification exponent emerges through the analysis of quantizing the amplitudes $\bx$.

For the converse proof, we take the $\rho^*_X, \rho^*_Y$ to minimize \eqref{eqn:EIDdiffVar}, and focus on the case where $\bX$ lies in a spherical shell with radius $\sqrt{n \rho_X^*\sigma^2_X}$ and small,  nonzero thickness. Then, the converse proceeds similar to that of Theorem \ref{thm:RIDdiffVar}, in the sense that the ``typical shell" is replaced by the new shell that depends on $\rho^*_X$.

\begin{proof}[Proof of Theorem~\ref{thm:EIDdiffVar}]\\
\textbf{Direct Part:} We will rely on the code construction given in the achievability proof of Theorem~\ref{thm:RIDdiffVar}, and hence we adopt the notation previously defined there.  To this end, let $(T,g)$ be the rate-$R$, $D$-admissible identification system defined in the achievability proof of Theorem~\ref{thm:RIDdiffVar}.  Recall that
\begin{align}
\angle\left( T\left( r_X \frac{\bx}{\|\bx\|} \right), \frac{\bx}{\|\bx\|} \right) \leq \theta_0,
\end{align}
where $\theta_0$ was defined as \eqref{eqn:def_ThetaZero}.

In a variation on the scheme used previously, we describe the amplitude $\|\bx\|$ by quantization as follows.
Let $\sigma^2_{\max}(n) \triangleq n\cdot \sigma_X^2$, and recall that $\eta$ was chosen to be a small positive constant.
Define the spherical shells $S^{(i)}$ as follows:
\begin{equation}
  S^{(i)} \triangleq S_{r^{(i)},r^{(i+1)}},
\end{equation}
where $r^{(i)} \triangleq \sqrt{n \cdot i\cdot \eta}$.

The modified signature assignment $T'$ then describes the ``direction" and ``amplitude" of $\bx$ as follows:
\begin{itemize}
\item If $\frac{1}{n}\|\bx\|^2 \leq \sigma_{\max}^2(n)$, then $T'(\bx) = \left(T\left( r_X \frac{\bx}{\|\bx\|} \right),i\right)$, where $i$ is chosen to satisfy $\bx \in S^{(i)}$.
\item If $\frac{1}{n}\|\bx\|^2 > \sigma_{\max}^2(n)$, then the signature $T'(\bx)$ is defined to be the erasure symbol $\ERASURE$.
\end{itemize}

The overall rate of the modified signature assignment $T'$ described above is $R$ (i.e., the rate of $T(\cdot)$), plus an additional $\frac{1}{n}\log \frac{\sigma^2_{\max}(n)}{\eta} = O(\frac{\log n}{n})$ (required for the quantization of $\|\bx\|$), and therefore remains essentially unchanged.  Therefore, the upper bound \eqref{upperBoundRThm1} also upper bounds the rate of the modified signature assignment function.  Let $g'$ be the optimal query function corresponding to $T'$ (defined in an analogous manner to \eqref{eqn:opt_decoder}).

Thus, we only need to analyze the  exponent attained by the proposed scheme.  To this end, let  $Z$ be a Chi-square random variable with $n$ degrees of freedom. The pdf of $Z$ is given by
\begin{equation}\label{eqn:Zpdf}
  f_Z(z) = \frac{z^{\frac{n}{2}-1}e^{-\frac{z}{2}}}{2^{n/2}\Gamma(\tfrac{n}{2})},
\end{equation}
where $\Gamma$ in \eqref{eqn:Zpdf} is the usual Gamma function, and should not be confused with the set-expansion operator $\Gamma^D$ defined previously.  Now, define  the random variables $Z_X \triangleq \frac{1}{\sigma_X^2}\|\bX\|^2$ and $Z_Y \triangleq \frac{1}{\sigma_Y^2}\|\bY\|^2$. Note that both $Z_X$ and $Z_Y$ are distributed according to \eqref{eqn:Zpdf}.  In order to proceed, we require the following lemma.

\begin{lem}\label{lem:supex}
    The probability $\Pr\left\{\frac{1}{n}\|\bX\|^2 > \sigma^2_{\max}(n)\right\}$ vanishes super-exponentially with $n$.
\end{lem}
\begin{proof} Appendix~\ref{app:supex}.\end{proof}

\bigskip

Now, we are in a position to analyze $\Pr\left\{\MAYBE\right\}$, where we again employ the shorthand notation $ \Pr\{\MAYBE\} \triangleq \Pr \{ g'(T'(\bX),\bY) = \MAYBE \}$ to simplify the presentation.
\begin{align}
  \Pr\{\MAYBE\} & \leq \Pr\left\{\MAYBE , \frac{1}{n}\|\bX\|^2 \leq \sigma^2_{\max}(n),\frac{1}{n}\|\bY\|^2 \leq \sigma^2_{\max}(n)\right\}\nonumber \\
&+ \Pr\left\{\frac{1}{n}\|\bX\|^2 > \sigma^2_{\max}(n)\right\} + \Pr\left\{\frac{1}{n}\|\bY\|^2 > \sigma^2_{\max}(n)\right\}.
\end{align}

By Lemma~\ref{lem:supex}, and a similar argument for $\Pr\left\{\frac{1}{n}\|\bY\|^2 > \sigma^2_{\max}(n)\right\}$, the last two terms of the above expression vanish super-exponentially and do not affect the exponent of $\Pr\{\MAYBE\}$. We therefore concentrate on the first term.

We can now write
\begin{align}
& \Pr\left\{\MAYBE , \tfrac{1}{n}\|\bX\|^2 \leq
\sigma^2_{\max}(n),\tfrac{1}{n}\|\bY\|^2 \leq \sigma^2_{\max}(n)\right\}\\
&= \Pr\left\{\MAYBE , Z_X \leq
n^2,Z_Y \leq \tfrac{\sigma_Y^2}{\sigma_X^2}n^2\right\}\\
&= \int_0^{n^2} \int_0^{\tfrac{\sigma_Y^2}{\sigma_X^2}n^2} \Pr\left\{\MAYBE \mid Z_X =
z_X,Z_Y = z_Y\right\} f_Z(z_X) f_Z(z_Y) dz_Ydz_X\\
&\leq \tfrac{\sigma_Y^2}{\sigma_X^2} n^4 \max_{
\underset{0 \leq z_Y \leq n^2 \sigma_Y^2/\sigma_X^2 }{0 \leq z_X \leq n^2}} \Pr\left\{\MAYBE \mid Z_X =
z_X,Z_Y = z_Y\right\} f_Z(z_X) f_Z(z_Y)\\
&\leq \tfrac{\sigma_Y^2}{\sigma_X^2} n^4 \max_{0 \leq \rho_X,\rho_Y } \Pr\left\{\MAYBE \mid Z_X =
n \rho_X,Z_Y = n \rho_Y\right\} f_Z(n \rho_X) f_Z(n \rho_Y)
\label{eqn:exponentUpperBound},
\end{align}
where $\rho_X \triangleq z_X/n$ and $\rho_Y \triangleq z_Y/n$.

\bigskip

The event $\{\MAYBE\}$ coincides with the event $\{\bY \in \Gamma^D(T'^{-1}(T'(\bX)))\}$. Let $\bU = T(r_X \bX / \|\bX\|)$, and observe that if $\tfrac{1}{n}\|\bX\|^2 \leq
\sigma^2_{\max}(n)$, then
\begin{align}
  \Gamma^D(T'^{-1}(T'(\bX))) &\subseteq  \Gamma^D(\CAP_{r^{(i)},r^{(i+1)}}(\bU,\theta_0))\label{inclusionA} \\
   &\subseteq \Gamma^{D'}(\CAP_{\|\bX\|}(\bU,\theta_0))),\label{inclusionB}
\end{align}
where \eqref{inclusionA} follows from similar arguments leading to \eqref{eqn:cap_contained}, and \eqref{inclusionB} follows with $D' \triangleq (\sqrt{D} + \sqrt{\eta})^2$.
We therefore continue with
\begin{align}
  &\Pr\left\{\MAYBE \mid Z_X = n \rho_X,Z_Y = n \rho_Y\right\} \nonumber\\
  & \leq \Pr\left\{\bY \in \Gamma^{D'}(\CAP_{\|\bX\|}(\bU,\theta_0)))\mid \frac{1}{n}\|\bX\|^2 = \rho_X \sigma_X^2 ,\frac{1}{n}\|\bY\|^2 = \rho_Y \sigma_Y^2\right\}\nonumber\\
  & = \Pr\left\{\bY \in \Gamma^{D'}(\CAP_{\sqrt{n \rho_X\sigma_X^2}}(\bU,\theta_0)))\mid \frac{1}{n}\|\bY\|^2 = \rho_Y \sigma_Y^2\right\} \label{bySphereSymm}\\
  & = \left\{
  \begin{array}{ll}
    0                            & \hbox{if $|\sqrt{\rho_X\sigma_X^2}-\sqrt{\rho_Y\sigma_Y^2}|\geq \sqrt{D'}$} \\
    1                            & \hbox{if $\rho_X\sigma_X^2+\rho_Y\sigma_Y^2 \leq D'$} \\
     \Omega(\theta_0 + \theta_1')  & \hbox{otherwise.}\label{eqn:OmegaPreAsym}
  \end{array}
\right.
\end{align}
where \eqref{bySphereSymm} follows by spherical symmetry of the Gaussian distribution, and
\begin{equation}\label{eqn:theta1defExponents1}
  \theta_1' \triangleq \arccos \frac{\rho_x\sigma_X^2 + \rho_Y\sigma_Y^2 - D'}{2\sqrt{\rho_X\sigma_X^2\cdot\rho_Y\sigma_Y^2}}.
\end{equation}
The identity \eqref{eqn:theta1defExponents1} follows from the law of cosines. %
The geometric image now is similar to that depicted in Fig.~\ref{fig:LawOfCosines_rXrY},
where here $r_X \triangleq \sqrt{n\sigma^2\rho_X}$ and $r_Y \triangleq \sqrt{n\sigma^2\rho_Y}$ denote
the actual radii of the vectors $\bX$ and $\bY$ (as opposed to their average value in the proof of Theorem~\ref{thm:RIDdiffVar}).

Next, using the bound \eqref{eqn:OmegaUpper} we have
\begin{equation}\label{logOmegaEqn}
  \frac{1}{n}\log \frac{1}{\Omega(\theta)} \geq
\left\{
  \begin{array}{ll}
    -\log\sin\theta + \tfrac{c}{n}\log n, & \hbox{$ 0 < \theta < \arccos(1/\sqrt{n})$;} \\
    0, & \hbox{otherwise.}
  \end{array}
\right.
\end{equation}
where $c$ is a universal constant.

Combined with \eqref{logOmegaEqn}, we compactly write the exponent corresponding to expression \eqref{eqn:OmegaPreAsym} as
\begin{align}
  &\bE_\Omega(\theta_0,D',\sigma_X^2,\sigma_Y^2,\rho_X,\rho_Y) \triangleq\nonumber\\
 &= \left\{
  \begin{array}{ll}
    \infty,                            & \hbox{if $|\sqrt{\rho_X\sigma_X^2}-\sqrt{\rho_Y\sigma_Y^2}|\geq \sqrt{D'}$} \\
    0,                            & \hbox{if $\rho_X\sigma_X^2+\rho_Y\sigma_Y^2 \leq D'$} \\
    -\log\sin\min\left[\tfrac{\pi}{2},\theta_0 + \theta_1'\right],& \hbox{otherwise.}
  \end{array}
\right.\label{eqn:OmegaAsym}
\end{align}
with $\theta_1'$ given in  \eqref{eqn:theta1defExponents1}.

Before we plug the above result into \eqref{eqn:exponentUpperBound}, we note that by Stirling's approximation we may write, for any fixed $\rho>0$:
\begin{align}
  f_Z(n\rho) &= \frac{1}{n\rho} \left(\frac{n\rho}{2}\right)^{n/2} \exp(-n \rho /2) \frac{1}{\Gamma(n/2)}\nonumber\\
&= \frac{1}{n\rho} \left(\frac{n\rho}{2}\right)^{n/2} \exp(-n \rho /2) \frac{1}{\sqrt\frac{4\pi}{n}\left(\frac{n}{2e}\right)^{n/2}} \left(1+O\left(\tfrac{1}{n}\right)\right)\nonumber\\
&= \exp\left[-n \left(\frac{\rho}{2} - \frac{1}{2} -\frac{1}{2}\log \rho \right)\right]
\frac{1}{\rho\sqrt{4\pi n}}  \left(1+O\left(\tfrac{1}{n}\right)\right)\nonumber\\
&\leq 2^{-n \bE_Z(\rho)} \cdot n^c,\label{eqn:ZpdfAsym}
\end{align}
where $\bE_Z(\cdot)$ was defined in \eqref{eqn:EZ} and $c$ is a universal constant.

Finally, we plug \eqref{eqn:OmegaAsym} and \eqref{eqn:ZpdfAsym} into the upper bound \eqref{eqn:exponentUpperBound}
on the (conditional) probability for $\MAYBE$ and conclude that the following exponent is achievable:

\begin{equation}
  \min_{\rho_X,\rho_Y\geq 0} \bE_Z(\rho_X)+\bE_Z(\rho_Y)
+\bE_\Omega(\theta_0,D',\sigma_X^2,\sigma_Y^2,\rho_X,\rho_Y).
\end{equation}

Since $\eta$ is arbitrarily small we may replace $D'$ with $D$ in the above.
We may therefore rewrite the achievable exponent as
\begin{equation}\label{eqn:EOmega}
  \min_{\rho_X,\rho_Y \geq 0} \bE_Z(\rho_X)+\bE_Z(\rho_Y) +
\bE_\Omega\left(\arcsin(2^{-R}),D,\sigma^2,\rho_X,\rho_Y\right).
\end{equation}

In order to simplify matters further, note that in \eqref{eqn:EOmega}, the minimizing $(\rho_X,\rho_Y)$ must satisfy:
\begin{align}
  \left|\sqrt{\rho_X\sigma_X^2}-\sqrt{\rho_Y\sigma_Y^2}\right| &< \sqrt{D}\label{eqn:optRegionCond1}\\
  \rho_X\sigma_X^2+\rho_Y\sigma_Y^2 &\geq D.\label{eqn:optRegionCond2}
\end{align}
The condition \eqref{eqn:optRegionCond1} must hold because otherwise the term $\bE_\Omega$ is infinite [see \eqref{eqn:OmegaAsym}].

To prove that \eqref{eqn:optRegionCond2} must hold, assume, for contradiction, that \eqref{eqn:EOmega} is minimized for $(\rho_X^*,\rho_Y^*)$ that satisfy
\begin{equation}
  \rho_X^*\sigma_X^2+\rho_Y^*\sigma_Y^2 < D.
\end{equation}
In this case, the value of \eqref{eqn:EOmega} at the minimizing point is $\bE_Z(\rho^*_X)+\bE_Z(\rho^*_Y)$.
If, say $\rho_X^*>1$, then we may replace it with another value $0<\rho_X^{**}<1$ that satisfies $\bE_Z(\rho_X^{**}) = \bE_Z(\rho_X^{*})$ that is guaranteed to exist (see the definition of $\bE_Z(\cdot)$).
The same argument holds for $\rho_Y^*$, and therefore we may assume that in this case both $\rho_X^*,\rho_Y^* \in (0,1]$.
Next, since $\bE_Z(\rho)$ is monotone decreasing for $\rho \in (0,1)$, we may increase $\rho_X^*$ and $\rho_Y^*$, while still in $(0,1]^2$,
until \eqref{eqn:optRegionCond2} is met with an equality.
Since the value of the objective function decreases, we arrive at a contradiction, meaning that \eqref{eqn:optRegionCond2} must hold for any minimizing $\rho_X,\rho_Y$.

Therefore the achievable exponent can be simplified to the expression \eqref{eqn:EIDdiffVar} and the proof of the direct part is concluded.

\bigskip

\textbf{Converse Part:}
Let $\rho_X^*,\rho_Y^*$ denote the minimizers of \eqref{eqn:EIDdiffVar} (in light of the discussion above, we can assume without loss of generality that $\rho_X^*,\rho_Y^*$ satisfy \eqref{eqn:rhoConditions}).
The proof of the converse  proceeds by focusing on values of $\bX$ and $\bY$ that satisfy
$\frac{1}{n}\|\bX\|^2 \cong \rho_X^* \sigma_X^2$ and $\frac{1}{n}\|\bY\|^2 \cong \rho_Y^* \sigma_Y^2$.
The details are as follows:

Let $0<\eta<\min(\rho_X^*,\rho_Y^*)$ be a small but fixed value. Define the following spherical caps:
\begin{equation}
  S_X^* \triangleq S_{r_X^-,r_X^+},\quad S_Y^* \triangleq S_{r_Y^-,r_Y^+},
\end{equation}
where $r_X^\pm \triangleq \sqrt{n \sigma_X^2(\rho_X^* \pm \eta)}$ and $r_Y^\pm \triangleq \sqrt{n \sigma_Y^2(\rho_Y^* \pm \eta)}$.

We then write the following:
\begin{align}
  \Pr\{\MAYBE\}
  & \geq \Pr\left\{\MAYBE ,\bX \in S_X^*, \bY \in S_Y^*\right\}\nonumber \\
  & = \Pr\left\{\MAYBE | \bX \in S_X^*, \bY \in S_Y^*\right\} \cdot \Pr \{\bX \in S_X^*\}\cdot \Pr \{\bY \in S_Y^*\}.\label{eqn:PrMaybeLBexp}
\end{align}

Consider the term $\Pr\{\bX \in S_X^*\}$:
\begin{align}
  \Pr\left\{\bX \in S_X^* \right\}
& = \Pr\left\{\frac{1}{n \sigma_X^2}\|\bX\|^2 \in (\rho_X^*-\eta,\rho_X^*+\eta)\right\}  \\
& = \int_{n(\rho_X^*-\eta)}^{n(\rho_X^*+\eta)} f_Z(z)dz \\
& \geq 2n\eta \min_{z \in [n(\rho_X^*-\eta),n(\rho_X^*+\eta)]} f_Z(z)\\
& \geq 2n\eta \cdot n^c \cdot 2^{-n \max_{\rho_X \in [\rho_X^*-\eta,\rho_X^*+\eta]} \bE_Z(\rho_X )},
\label{eqn:PrXinSXstar}
\end{align}
where \eqref{eqn:PrXinSXstar} follows from Stirling's approximation  similar to \eqref{eqn:ZpdfAsym}.  A similar derivation applies for $\Pr\{\bY \in S_Y^*\}$. Thus, it follows from \eqref{eqn:PrMaybeLBexp} and continuity of $\bE_Z(\cdot)$ that
\begin{align}
&-\frac{1}{n}\log   \Pr\{\MAYBE\} \notag\\
&\leq -\frac{1}{n}\log\left[ \Pr\left\{\MAYBE | \bX \in S_X^*, \bY \in S_Y^*\right\} \right]+ \bE_Z(\rho_X^* ) + \bE_Z(\rho_Y^* ) + \eta'+ O\left(\frac{\log n}{n} \right), \label{exponentUpperBound}
\end{align}
where $\eta'$ is a quantity tending to zero as $\eta\ra0$.

We now concentrate on the term $\Pr\left\{\MAYBE | \bX \in S_X^*, \bY \in S_Y^*\right\}$, and proceed in a manner similar to the converse proof of Theorem~\ref{thm:RIDdiffVar}.  To this end, let $T:\Reals^n\ra\{1,...,2^{nR}\}$ denote the signature assignment for the scheme at hand.
Define the mapping $\tilde T:S_X^* \ra \{1,...,2^{nR}\}$ as $\tilde T(\bx) = T(\bx)$ for all $\bx\in S_X^*$.  That is,  $\tilde T(\cdot)$ is the restriction of $T(\cdot)$  to  $S_X^*$. Let $\tilde T^{-1}(\cdot)$ denote the inverse mapping of $\tilde T(\cdot)$:
\begin{align}
  \tilde T^{-1}(i)
  &\triangleq \{\bx \in S_X^* : T(\bx) = i\} \\
  &= T^{-1}(i) \cap S_X^*.
\end{align}

Let $p_i \triangleq \Pr\{\bX \in \tilde T^{-1}(i) | \bX \in S_X^*\}$, so $\sum_{i=1}^{2^{nR}}p_i = 1$.
Define $r_X \triangleq \sqrt{n \sigma_X^2\rho_X^*}$, and let the set $A_i \subseteq S_{r_X}$ denote the projection of $\tilde T^{-1}(i)$ onto the sphere $S_{r_X}$.  In other words,
\begin{equation}
  A_i = \left\{r_X\frac{\bx}{\|\bx\|}  : \bx \in \tilde T^{-1}(i)  \right\}.
\end{equation}

Let $\alpha_i$ denote the fraction of the surface area of $S_{r_X}$ that is occupied by $A_i$.
By the spherical symmetry of the distribution of $\bX$, $\alpha_i$ is also equal to the probability that the projection of $\bX$ onto $S_{r_X}$ lies in $A_i$.
Therefore $\alpha_i \geq p_i$, with equality if and only if $\tilde T^{-1}(i)$ is a thick cap with inner and outer radii $r_X^-$ and $r_X^+$ respectively.

Let $D' \triangleq (\sqrt{D} +\sqrt{\sigma_X^2(\rho_X^*-\eta)} -  \sqrt{\sigma_X^2\rho_X^*})^2$.
As in \eqref{eqn:AiProperty} we have that
\begin{equation}\label{eqn:AiPropertyExp}
  \Gamma^{D'}\left(A_i\right) \subseteq \Gamma^D\left(\tilde T^{-1}(i)\right).
\end{equation}

Now let $D'' \triangleq (\sqrt{D'} +\sqrt{\sigma_Y^2(\rho_Y^*-\eta)} -  \sqrt{\sigma_Y^2\rho_Y^*})^2$,
and let the set $B_i\subseteq S_{r_Y}$ denote the $D''$-expansion of $A_i$, restricted to the sphere $S_{r_Y}$, where $r_Y \triangleq \sqrt{n \sigma_Y^2\rho_Y^*}$, i.e.
\begin{equation}
  B_i \triangleq \Gamma^{D''}\left(A_i\right) \cap S_{r_Y}.
\end{equation}
Let $\beta_i$ denote the fraction of $S_{r_Y}$ that is occupied by $B_i$. Let the set $C_i$ denote the $r_Y^-,r_Y^+$ thickening of $B_i$ as follows:
\begin{equation}
  C_i = \left\{\by\in S_Y^* :  r_Y \frac{\by}{\|\by\|} \in B_i  \right\}.
\end{equation}
As in \eqref{eqn:CiProperty} we have that
\begin{equation}\label{eqn:CiPropertyExp}
  C_i \subseteq \Gamma^{D'}\left(A_i\right).
\end{equation}

\bigskip

Suppose that $\bX = \bx \in S_X^*$ and that $T(\bx) = i$. Then we have, with the aid of \eqref{eqn:AiPropertyExp} and \eqref{eqn:CiPropertyExp}:
\begin{align}
  \Pr\{\MAYBE | \bX =\bx \in S_X^*,\bY \in S_Y^*\}
  &\geq \Pr\left\{\bY \in  \Gamma^D\left(T^{-1}(i)\right)|\bY \in S_Y^*\right\} \\
  &\geq \Pr\left\{\bY \in  \Gamma^D\left(\tilde T^{-1}(i)\right)|\bY \in S_Y^*\right\} \\
  &\geq \Pr\left\{\bY \in  \Gamma^{D'}\left(A_i\right)|\bY \in S_Y^*\right\} \\
  &\geq \Pr\left\{\bY \in C_i |\bY \in S_Y^*\right\}\\
  &=\beta_i,
\end{align}
where the last equality follows from the spherical symmetry of the pdf of $\bY$.

As in the proof of the converse of Theorem~\ref{thm:RIDdiffVar},
we apply the isoperimetric inequality on the sphere for the sets $A_i$ and $B_i$.
We get that the set $A_i^*$ that minimizes $\beta_i$ for given $\alpha_i$ is the set $\CAP_{r_X}(\bu,\theta_i)$,
where $\bu$ is an arbitrary point, $\theta_i \triangleq  \Omega^{-1}(\alpha_i)$, and $B_i^*$ is the set $\CAP_{r_Y}(\bu,\theta_i')$, defined by
\begin{align}
  \theta_i'
  &\triangleq \theta_i + \theta_{D''}
\end{align}
where
\begin{equation}\label{eqn:theta1defExponents}
  \theta_{D''} \triangleq \arccos \frac{\rho_X^*\sigma_X^2 + \rho_Y^*\sigma_Y^2 - D''}{2\sqrt{\rho_X^*\sigma_X^2\cdot\rho_Y^*\sigma_Y^2}}.
\end{equation}

Therefore the (normalized) surface area of $B_i^*$ is given by $\beta_i^* = \Omega(\theta_i')$. It follows that
\begin{equation}
  \Pr\{\MAYBE | \bX =\bx \in S_X^*,\bY \in S_Y^*\} \geq \Omega(\theta_{D''} + \Omega^{-1}(p_i)),
\end{equation}
and the average (conditional) probability $\Pr\{\MAYBE | \bX \in S_X^*,\bY \in S_Y^* \}$ is bounded by
\begin{align}
  &\Pr\{\MAYBE | \bX \in S_X^*,\bY \in S_Y^* \}\\
  &= \sum_{i=1}^{2^{nR}} \Pr\{T(\bX) = i | \bX \in S_X^*\} \Pr\{\MAYBE | T(\bX) = i , \bX \in S_X^*,\bY \in S_Y^* \}\\
  &\geq \sum_{i=1}^{2^{nR}} p_i \cdot
    \Omega\left(\theta_{D''}+\Omega^{-1}(p_i)\right).
\end{align}

Now, let $0 < c < 1$, and invoke Lemma~\ref{lem:Markov} %
to conclude that
\begin{equation}\label{eqn:pstar_for_exp}
  R \geq \frac{1}{n}\log \frac{1-c}{p^*},
\end{equation}
where $p^*$ is the solution to
\begin{equation}\label{eqn:pstar_and_E}
  \Omega(\theta_{D''} + \Omega^{-1}(p))= c^{-1}  \Pr\{\MAYBE | \bX \in S_X^*,\bY \in S_Y^*\}.
\end{equation}
Since $\Omega(\cdot)$ is monotone increasing, so is $\Omega^{-1}(\cdot)$.  Therefore, \eqref{eqn:pstar_for_exp} and \eqref{eqn:pstar_and_E} imply the inequality
\begin{align}
 \Pr\{\MAYBE | \bX \in S_X^*,\bY \in S_Y^*\} \geq c\cdot \Omega\left(\theta_{D''}+  \Omega^{-1}\left((1-c)2^{-nR}\right)\right). \label{pMaybLBc}
\end{align}
It is a straightforward exercise to verify (e.g., by Taylor series expansion) that
\begin{align}
\Omega^{-1}\left((1-c)2^{-nR}\right) = \arcsin\left(2^{-R}\right) + O\left(\tfrac{\log n}{n}\right).
\end{align}
If $\theta_{D''} + \arcsin\left(2^{-R}\right) \geq \pi/2$, then \eqref{pMaybLBc} and the definition of $\Omega(\cdot)$ yield
\begin{align}
 \Pr\{\MAYBE | \bX \in S_X^*,\bY \in S_Y^*\} \geq c/2,
\end{align}
which, combined with  \eqref{exponentUpperBound}, yields the desired  upper bound
\begin{align}
-\frac{1}{n}\log   \Pr\{\MAYBE\} \leq -\log \sin \left(\frac{\pi}{2}\right) + \bE_Z(\rho_X^* ) + \bE_Z(\rho_Y^* ) + \eta'+ O\left(\frac{\log n}{n} \right).\label{ub1Exp}
\end{align}
On the other hand, if $\theta_{D''} + \arcsin\left(2^{-R}\right) < \pi/2$, then the hypothesis of Lemma \ref{lem:Omega} is satisfied for $n$ sufficiently large, and the estimate \eqref{eqn:OmegaLower} gives
\begin{align}
&-\frac{1}{n}\log \Pr\{\MAYBE | \bX \in S_X^*,\bY \in S_Y^*\}
&\leq -\log \sin\left( \theta_{D''} + \Omega^{-1}\left((1-c)2^{-nR}\right)\right) + O\left(\frac{\log n }{n}\right). \label{ub2Exp}
\end{align}

By letting $\eta$ be arbitrarily small we can infer from \eqref{ub1Exp} and \eqref{ub2Exp} that any sequence of identification schemes $\{g^{(n)},T^{(n)}\}_{n\ra \infty}$ must satisfy
\begin{align*}
&\limsup_{n \ra \infty} -\frac{1}{n}\log  \Pr\{  g^{(n)}(T^{(n)}(\bX),\bY) =   \MAYBE\}  \\
 &\leq  \bE_Z(\rho_X^*) + \bE_Z(\rho_Y^*) -\log\sin\min\left[\tfrac{\pi}{2},\arcsin\left(2^{-R}\right)+
\arccos \frac{\rho_X^*\sigma_X^2 + \rho_Y^*\sigma_Y^2 - D}{2\sqrt{\rho_X^*\sigma_X^2\cdot\rho_Y^*\sigma_Y^2}}
\right],
\end{align*}
as desired.
\end{proof}

\begin{proof}[Proof of Corollary \ref{cor:EID}]
Let $\rho_X,\rho_Y$ satisfy \eqref{eqn:rhoConditions}.
We claim that  the quantity
\begin{align}
 \bE_Z(\rho_X)+\bE_Z(\rho_Y) +\wp(R,D,{\rho_X\sigma^2},{\rho_Y\sigma^2}) \label{eqn:EIDdiffVarCorProof}
\end{align}
can not increase if $\rho_X$ and $\rho_Y$ are both replaced by their average $\overline{\rho} := (\rho_X+\rho_Y)/2$, which continues to satisfy \eqref{eqn:rhoConditions}. To see that this is indeed the case, note that $\bE_Z(\cdot)$ is convex, and therefore Jensen's inequality implies
\begin{align}
\bE_Z(\rho_X)+\bE_Z(\rho_Y) \geq 2 \bE_Z(\overline{\rho}).
\end{align}
Next, the inequality of arithmetic and geometric means implies
\begin{align}
\frac{\rho_X \sigma^2 + \rho_Y \sigma_Y^2 - D}{2\sigma^2\sqrt{\rho_X \rho_Y}} \geq \frac{2 \overline{\rho} \sigma^2 - D}{2\overline{\rho} \sigma^2},
\end{align}
and therefore, since $\arccos(x)$ is monotone decreasing on $x\in[0,1]$,
\begin{align}
\arccos \frac{\rho_X \sigma^2 + \rho_Y \sigma_Y^2 - D}{2\sigma^2\sqrt{\rho_X \rho_Y}} \leq \arccos \frac{2 \overline{\rho} \sigma^2 - D}{2\overline{\rho} \sigma^2}. \label{eqn:arccosIneq}
\end{align}
Since $-\log \sin (x)$ is decreasing on $x\in[0,\pi/2]$, \eqref{eqn:arccosIneq} implies
\begin{align}
\wp(R,D,{\rho_X\sigma^2},{\rho_Y\sigma^2}) \geq \wp(R,D,{\overline{\rho}\sigma^2},{\overline{\rho}\sigma^2}),
\end{align}
which proves that \eqref{eqn:EIDdiffVarCorProof} can not increase if $\rho_X$ and $\rho_Y$ are both replaced by their average $\overline{\rho}$.  The observation that
\begin{align}
  2\bE_Z(\rho) +\wp(R,D,{\rho \sigma^2},{\rho \sigma^2})
\end{align}
is monotone increasing for $\rho>1$ completes the proof.
\end{proof}

\subsection{General Sources and the Extremal Property of the Gaussian}
The proof of Theorem \ref{thm:achGeneral} can be accomplished by restricting our attention to the setting where $X$ and $Y$ are discrete random variables.  Therefore, the usual typicality machinery will be useful to us, and we review a few facts before beginning the proof of Theorem \ref{thm:achGeneral}. We should also note that the method of types is used in the proofs in \cite{Ahlswede97}, but the proof here, which is similar in spirit, is significantly simpler and shorter, partially because we are only interested in the achievable rate (and not in the exponent). To this end, let $\mathcal{T}_{\epsilon}^{(n)}$ denote the usual $\epsilon$-typical set (cf. \cite[Chapter 2]{ElGamalYHKim2012}).  That is, we define the empirical pmf of $\mathbf{w}\in\mathcal{W}^n$ as
\begin{align}
\pi(w|\mathbf{w}) = \frac{|i:w_i = w|}{n} \mbox{~~for $w\in\mathcal{W}$,}
\end{align}
and, for $W\sim P_W$, the set of $\epsilon$-typical $n$-sequences is defined by
\begin{align}
\mathcal{T}_{\epsilon}^{(n)}(W) = \left\{\mathbf{w} : |\pi(w|\mathbf{w}) - P_W(w)| \leq \epsilon P_W(w)  \mbox{~for all $w\in\mathcal{W}$}\right\}.
\end{align}
Observe that if $\mathbf{W} \sim \prod_{i=1}^nP_W(w_i)$, then the union of events bound and Hoeffding's inequality imply
\begin{align}
\Pr\left\{ \mathbf{W} \notin \mathcal{T}_{\epsilon}^{(n)}(W) \right\} &\leq
\sum_{w\in \mathcal{W}} \Pr\left\{ |\pi(w|\mathbf{W}) - P_W(w)| > \epsilon P_W(w) \right\} \\
&\leq
\sum_{ \substack{ w\in \mathcal{W}: \\ P_W(w) > 0} } 2 \exp\left( -n\left(\epsilon P_W(w)\right)^2 \right).
\end{align}
Therefore, if $|\mathcal{W}|<\infty$,
\begin{align}
\Pr\left\{ \mathbf{W} \notin \mathcal{T}_{\epsilon}^{(n)}(W) \right\} \leq \exp\left(-n\delta(\epsilon) \right), \label{eqn:probNotTypical}
\end{align}
where $\delta(\epsilon)$ denotes a positive quantity satisfying $\lim_{\epsilon\rightarrow 0}\delta(\epsilon)=0$.

One useful fact is the so-called \emph{Typical Average Lemma} \cite[Section 2.4]{ElGamalYHKim2012}:
\begin{lem}[Typical Average Lemma]\label{lem:TypAvg}
If  $\mathbf{w} \in \mathcal{T}_{\epsilon}^{(n)}(W)$, then
\begin{align*}
(1-\epsilon)\EE [ f(W) ] \leq \frac{1}{n}\sum_{i=1}^ng(w_i) \leq (1+\epsilon)\EE [ f(W) ]
\end{align*}
for any nonnegative function $f(w)$ on $\mathcal{W}$.
\end{lem}

Now, we state a simple variant of the Covering Lemma  \cite[Lemma 3.3]{ElGamalYHKim2012}:
\begin{lem}\label{lem:CoveringVariant}
Let  $P_{WV}$ be a joint probability distribution on the finite alphabet $\mathcal{W} \times \mathcal{V}$, with corresponding marginals $P_W$ and $P_V$.  Let $\mathbf{W} \sim \prod_{i=1}^nP_W(w_i)$ and let $\mathbf{V}(m)$, $m\in \{1,2,\dots,2^{nR}\}$, be random sequences, independent of each other and of $\mathbf{W}$, each distributed according to $\prod_{i=1}^nP_V(v_i)$.  Then, for $n$ sufficiently large, there exists positive functions $\delta(\epsilon),\tilde{\delta}(\epsilon)$ satisfying $\lim_{\epsilon\rightarrow 0}\delta(\epsilon)=\lim_{\epsilon\rightarrow 0}\tilde{\delta}(\epsilon)=0$ and
\begin{align*}
\Pr\left\{ (\mathbf{W},\mathbf{V}(m)) \notin \mathcal{T}_{\epsilon}^{(n)}(W,V) \mbox{~for all $m$}\right\} \leq \exp\left(-n\delta(\epsilon) \right) + \exp\left(-2^{n(R-I(W;V)-\tilde{\delta}(\epsilon))} \right).
\end{align*}
\end{lem}
\begin{proof}
The proof follows that of \cite[Lemma 3.3]{ElGamalYHKim2012} verbatim, invoking \eqref{eqn:probNotTypical} where appropriate.
\end{proof}

We require one more result before moving on to the proof of Theorem \ref{thm:achGeneral}.
\begin{lem}\label{lem:scalingDistance}
Let $P_W$ and $P_V$ be probability distributions with finite second moments $\sigma_W^2$ and $\sigma_V^2$, respectively.  If $\mathbf{w} \in \mathcal{T}_{\epsilon}^{(n)}(W)$,  $\mathbf{v} \in \mathcal{T}_{\epsilon}^{(n)}(V)$, and $\frac{1}{n}\|\mathbf{w}-\mathbf{v}\|^2 \leq D$, then
\begin{align}
\frac{1}{n}\left\|\sqrt{\frac{\sigma_V}{\sigma_W}}\mathbf{w}-\sqrt{\frac{\sigma_W}{\sigma_V}}\mathbf{v}\right\|^2 \leq D-(\sigma_W-\sigma_V)^2+\epsilon|\sigma_W^2-\sigma_V^2|.
\end{align}
\end{lem}
\begin{proof}
Without loss of generality, assume $\sigma_V\geq \sigma_W$.  Note that the assumption $\frac{1}{n}\|\mathbf{w}-\mathbf{v}\|^2 \leq D$ implies
\begin{align}
-\frac{2}{n}\mathbf{w}^T\mathbf{v} \leq D-\frac{1}{n}\|\mathbf{w}\|^2-\frac{1}{n}\|\mathbf{v}\|^2.
\end{align}
Moreover,  Lemma \ref{lem:TypAvg} implies the following inequalities
\begin{align}
\frac{1}{n}\|\mathbf{w}\|^2 &\leq (1+\epsilon)\sigma_W^2\\
\frac{1}{n}\|\mathbf{v}\|^2 &\geq (1-\epsilon)\sigma_V^2.
\end{align}
Therefore, it follows that
\begin{align}
\frac{1}{n}\left\|\sqrt{\frac{\sigma_V}{\sigma_W}}\mathbf{w}-\sqrt{\frac{\sigma_W}{\sigma_V}}\mathbf{v}\right\|^2
&=\frac{1}{n}\left(\frac{\sigma_V}{\sigma_W}\|\mathbf{w}\|^2 + \frac{\sigma_W}{\sigma_V}\|\mathbf{v}\|^2 - 2 \mathbf{w}^T\mathbf{v}\right)\\
&\leq D + \frac{1}{n}\left(\left(\frac{\sigma_V}{\sigma_W}-1\right)\|\mathbf{w}\|^2 + \left(\frac{\sigma_W}{\sigma_V}-1\right)\|\mathbf{v}\|^2\right)\\
&\leq D + (1+\epsilon)\sigma_W^2 \left(\frac{\sigma_V}{\sigma_W}-1\right) + (1-\epsilon)\sigma_V^2\left(\frac{\sigma_W}{\sigma_V}-1\right)\\
&= D-(\sigma_W-\sigma_V)^2+\epsilon(\sigma_V^2-\sigma_W^2).
\end{align}
Considering the symmetric case where $\sigma_V \leq \sigma_W$ gives
\begin{align}
\frac{1}{n}\left\|\sqrt{\frac{\sigma_V}{\sigma_W}}\mathbf{w}-\sqrt{\frac{\sigma_W}{\sigma_V}}\mathbf{v}\right\|^2
&\leq D-(\sigma_W-\sigma_V)^2+\epsilon(\sigma_W^2-\sigma_V^2),
\end{align}
completing the proof.
\end{proof}

\begin{proof}[Proof  of Theorem \ref{thm:achGeneral}]
We can assume that $X$ and $Y$ are discrete random variables with finite alphabet $\mathcal{X}\subset \mathbb{R}$. The extension to continuous distributions with finite second moments follows by the usual quantization arguments and continuity of $\|\cdot\|$. Fix $\epsilon>0$ and a conditional pmf $P_{\hat{X}|X}(\hat{x}|x)$, where the alphabet $\hat{\mathcal{X}}$ is an arbitrary subset of $\Reals$ with finite support.  Throughout, the random variables $(Y,X,\hat{X})$ are drawn according to the joint distribution \begin{align}
P_{YX\hat{X}}(y,x,\hat{x}) = P_Y(y) P_{X\hat{X}}(x,\hat{x}) = P_Y(y) P_X(x) P_{\hat{X}|X}(\hat{x}|x).
\end{align}

\textbf{Random signature assignment.}
Randomly and independently generate $2^{nR}$ sequences $\hat{\bx}(t), t\in\{1,2,\dots,2^{nR}\}$, each according to $\prod_{i=1}^n P_{\hat{X}}(\hat{x}_i)$. Given a sequence $\bx$, find an index $t$ such that $(\bx,\hat{\mathbf{x}}(t))\in\mathcal{T}_{\epsilon}^{(n)}(X,\hat{X})$ and put $T(\mathbf{x})=t$.  If there is more than one such index, break ties arbitrarily.    If there is no such index, put $T(\mathbf{x})=\ERASURE$.  Observe that the rate $R$ is negligibly affected by the addition of the additional ``erasure" signature $\ERASURE$ (as in the proofs of Theorems \ref{thm:RIDdiffVar} and \ref{thm:EIDdiffVar}).

\textbf{Definition of the query function.}  %
In order to simplify notation, define the quantity
\begin{align}
\Psi \triangleq \sqrt{(1+\epsilon)\mathbb{E}\left[\left(\sqrt{\frac{\sigma_Y}{\sigma_X}}X - \hat{X}\right)^2\right]}
+\sqrt{D-(\sigma_X-\sigma_Y)^2+\epsilon|\sigma_X^2-\sigma_Y^2|}. \label{eqn:PsiDef}
\end{align}

For a signature $t\in\{1,2,\dots,2^{nR}\}\cup\{\ERASURE\}$ and a sequence $\by$, define
\begin{align*}
g(t,\by)=\left\{ \begin{array}{ll}
\MAYBE & \mbox{if}  \left\{ \begin{array}{l}
\mathbf{y} \notin \mathcal{T}_{\epsilon}^{(n)}(Y), \mbox{~or}  \phantom{\frac{~}{~}} \\
t=\ERASURE, \mbox{~or} \phantom{\frac{~}{~}}\\
\frac{1}{\sqrt{n}}\left\| \sqrt{\frac{\sigma_X}{\sigma_Y}}\mathbf{y} - \hat{\mathbf{x}}(t)\right\|  \leq \Psi \mbox{~and $t\neq \ERASURE$}
\end{array}\right. \\
\NO & \mbox{otherwise} .
\end{array}\right.
\end{align*}

\textbf{Scheme analysis.}
First, we check to ensure that $g(\cdot,\cdot)$ does not produce any false negatives; that is, we need to verify that $(T,g)$ is $D$-admissible. Note that $g(T(\bx),\by)$ returns $\MAYBE$ if  $\mathbf{y} \notin \mathcal{T}_{\epsilon}^{(n)}(Y)$  or $T(\bx)=\ERASURE$.
Therefore, we only need to show that $g(T(\bx),\by)$ returns $\MAYBE$ if $\mathbf{y} \in \mathcal{T}_{\epsilon}^{(n)}(Y)$, $(\mathbf{x},\hat{\mathbf{x}}(T(\bx)))\in\mathcal{T}_{\epsilon}^{(n)}(X,\hat{X})$, and $\frac{1}{{n}}\| \mathbf{x} - {\mathbf{y}}\|^2 \leq {D}$.

Under these assumptions, note that Lemma \ref{lem:TypAvg} implies
\begin{align}
\frac{1}{{n}}\left\| \sqrt{\frac{\sigma_Y}{\sigma_X}} \mathbf{x} - \hat{\mathbf{x}}(t)\right\|^2 \leq {(1+\epsilon)\mathbb{E}\left[\left(\sqrt{\frac{\sigma_Y}{\sigma_X}}X - \hat{X}\right)^2\right]}.\label{eqn:typicalAverage}
\end{align}
Next, recall that $(\mathbf{x},\hat{\mathbf{x}}(T(\bx)))\in\mathcal{T}_{\epsilon}^{(n)}(X,\hat{X})$ implies $\mathbf{x}\in\mathcal{T}_{\epsilon}^{(n)}(X)$.  Hence, under the assumption that $\frac{1}{{n}}\| \mathbf{x} - {\mathbf{y}}\|^2 \leq {D}$, Lemma \ref{lem:scalingDistance} implies
\begin{align}
\frac{1}{{n}}\left\|\sqrt{\frac{\sigma_Y}{\sigma_X}}\mathbf{x}-\sqrt{\frac{\sigma_X}{\sigma_Y}}\mathbf{y}\right\|^2 \leq {D-(\sigma_X-\sigma_Y)^2+\epsilon|\sigma_X^2-\sigma_Y^2|}.\label{eqn:scaledDistance}
\end{align}
Combining the triangle inequality, \eqref{eqn:typicalAverage}, \eqref{eqn:scaledDistance}, and \eqref{eqn:PsiDef}, we have
\begin{align}
\frac{1}{\sqrt{n}}\left\| \sqrt{\frac{\sigma_X}{\sigma_Y}}\mathbf{y} - \hat{\mathbf{x}}(t)\right\| &\leq
\frac{1}{\sqrt{n}}\left\| \sqrt{\frac{\sigma_Y}{\sigma_X}} \mathbf{x} - \hat{\mathbf{x}}(t)\right\|
+\frac{1}{\sqrt{n}}\left\|\sqrt{\frac{\sigma_Y}{\sigma_X}}\mathbf{x}-\sqrt{\frac{\sigma_X}{\sigma_Y}}\mathbf{y}\right\|
\leq \Psi.
\end{align}
Hence, $g(T(\bx),\by)$ returns $\MAYBE$ if $\mathbf{y} \in \mathcal{T}_{\epsilon}^{(n)}(Y)$, $(\mathbf{x},\hat{\mathbf{x}}(T(\bx)))\in\mathcal{T}_{\epsilon}^{(n)}(X,\hat{X})$, and $\frac{1}{{n}}\| \mathbf{x} - {\mathbf{y}}\|^2 \leq {D}$.  Therefore, $(T,g)$ is $D$-admissible as desired.

Next, we check to ensure that $\Pr\{g(T(\mathbf{X}),\mathbf{Y})= \MAYBE\}$ is small.  To this end, consider the events
\begin{align*}
&\mathcal{E}_0 = \left\{\mathbf{y} \notin \mathcal{T}_{\epsilon}^{(n)}(Y) \right\},\\
&\mathcal{E}_1 = \left\{T(\mathbf{X})=\ERASURE \right\},\\
&\mathcal{E}_2 = \left\{ \frac{1}{\sqrt{n}}\left\| \sqrt{\frac{\sigma_X}{\sigma_Y}}\mathbf{Y} - \hat{\mathbf{X}}(T(\bX))\right\| \leq  \Psi
 \right\},
\end{align*}
and observe that $\Pr\{g(T(\mathbf{X}),\mathbf{Y})= \MAYBE\} \leq \Pr\{\mathcal{E}_0\}+\Pr\{\mathcal{E}_1\}+\Pr\{\mathcal{E}_2\}$ by the union of events bound.

We have already seen in \eqref{eqn:probNotTypical} that
\begin{align}
 \Pr\{\mathcal{E}_0\} \leq \exp\left(-n\delta(\epsilon) \right)
\end{align}
for some positive $\delta(\epsilon)$ satisfying $\lim_{\epsilon\rightarrow 0} \delta(\epsilon)=0$.

Next, Lemma \ref{lem:CoveringVariant} implies that, for $n$ sufficiently large,
\begin{align}
\EE_T \left[\Pr\{\mathcal{E}_1\}\right] \leq \exp\left(-n\delta(\epsilon) \right) + \exp\left(-2^{n(R-I(X;\hat{X})-\tilde{\delta}(\epsilon))} \right),
\end{align}
where $\EE_T \left[\Pr\{\mathcal{E}_1\}\right]$ denotes the value of $\Pr(\mathcal{E}_1)$ averaged over the random choice of the signature assignment $T(\cdot)$.

Let  $\hat{\mathbf{X}}$ be distributed according to $\prod_{i=1}^n P_{\hat{X}}(\hat{x}_i)$, independent of $\mathbf{Y}\sim \prod_{i=1}^n P_{{Y}}(y_i)$.  An application of Hoeffding's inequality implies
\begin{align}
\Pr\left(\frac{1}{\sqrt{n}}\left\| \sqrt{\frac{\sigma_X}{\sigma_Y}}\mathbf{Y} - \hat{\mathbf{X}}\right\|  \leq \sqrt{\mathbb{E}\left[\left(\sqrt{\frac{\sigma_X}{\sigma_Y}}Y-\hat{X}\right)^2\right]} -\epsilon \right)\leq \exp(-n\delta(\epsilon)). \label{eqn:E2upperbound}
\end{align}
Since the sequence $\mathbf{Y}$ is independent of $\mathbf{X}$, and is therefore  also independent of $ \hat{\mathbf{X}}(T(\mathbf{X}))$, \eqref{eqn:E2upperbound} implies that
\begin{align}
\EE_T \left[\Pr\{\mathcal{E}_2\}\right] \leq \exp\left(-n\delta(\epsilon) \right)
\end{align}
if
\begin{align}
\Psi \leq
\sqrt{\mathbb{E}\left[\left(\sqrt{\frac{\sigma_X}{\sigma_Y}}Y-\hat{X}\right)^2\right]} -\epsilon.\label{eqn:thmConditions}
\end{align}
Therefore, if \eqref{eqn:thmConditions} holds, we have
\begin{align}
\EE_T \left[\Pr\{g(T(\mathbf{X}),\mathbf{Y})= \MAYBE\}\right] \leq 3\exp\left(-n\delta(\epsilon) \right) + \exp\left(-2^{n(R-I(X;\hat{X})-\tilde{\delta}(\epsilon))} \right),
\end{align}
implying the existence of a sequence of $D$-admissible, rate $R > I(X;\hat{X})$ schemes for which $\Pr\{g(T(\mathbf{X}),\mathbf{Y})= \MAYBE\}$ is exponentially small in $n$.  Since $\epsilon$ was arbitrary, the proof is complete.
\end{proof}

\begin{proof}[Proof of Theorem~\ref{thm:GaussianExtreme}]  Since $d(\cdot,\cdot)$ is translation invariant, we can assume without loss of generality that $P_X$ and $P_Y$ have mean zero.
Also, note that it is sufficient to consider $D$ in the interval $(\sigma_X-\sigma_Y)^2 \leq D \leq \sigma_X^2+\sigma_Y^2$.  Indeed, if $D > \sigma_X^2+\sigma_Y^2$, then the theorem asserts that $R_\ID(D,P_X,P_Y) \leq \infty$, which is trivially true.  On the other hand, we can argue that $R_\ID(D,P_X,P_Y)=0$ for $D< (\sigma_X-\sigma_Y)^2$ by monotonicity of $R_\ID(D,P_X,P_Y)$ in $D$ and the fact that $R_\ID((\sigma_X-\sigma_Y)^2,P_X,P_Y) =0$.

Therefore, assume $(\sigma_X-\sigma_Y)^2 < D \leq \sigma_X^2+\sigma_Y^2$ and consider the conditional distribution $P_{\hat{X}|X}$ defined by $\hat{X}=\rho \sqrt{\frac{\sigma_Y}{\sigma_X}}X + Z$, where $Z\sim N(0,\sigma_Z^2)$ is independent of $X$ and $\rho,\sigma_Z^2$ are given by
\begin{align}
&\rho = \frac{(\sigma_X+\sigma_Y)^2-D}{(2\sigma_X\sigma_Y)} ~~
&\sigma_Z^2 =\frac{((\sigma_X+\sigma_Y)^2-D)(\sigma_X^2+\sigma_Y^2-D)^2}{4\sigma_X\sigma_Y (D-(\sigma_X-\sigma_Y)^2)}.
\end{align}
With $P_{\hat{X}|X}$ defined in this way, the following identities are readily verified
\begin{align}
\sqrt{\mathbb{E}\left[\left(\sqrt{\frac{\sigma_X}{\sigma_Y}}Y-\hat{X}\right)^2\right]}
&= \sqrt{\sigma_X\sigma_Y(1 + \rho^2) + \sigma_Z^2} =  \frac{2\sigma_X\sigma_Y}{\sqrt{D-(\sigma_X-\sigma_Y)^2}} \label{eqn:testIdent1}\\
 \sqrt{\mathbb{E}\left[\left(\sqrt{\frac{\sigma_Y}{\sigma_X}}X - \hat{X}\right)^2\right]}
 &= \sqrt{\sigma_X\sigma_Y(1 -\rho)^2 + \sigma_Z^2} =\frac{\sigma_X^2+\sigma_Y^2-D}{\sqrt{D-(\sigma_X-\sigma_Y)^2}}.\label{eqn:testIdent2}
\end{align}
Therefore, \eqref{eqn:testIdent1} and \eqref{eqn:testIdent2} yield the identity
\begin{align}
 \sqrt{\mathbb{E}\left[\left(\sqrt{\frac{\sigma_X}{\sigma_Y}}Y - \hat{X}\right)^2\right]}  = \sqrt{\mathbb{E}\left[\left(\sqrt{\frac{\sigma_Y}{\sigma_X}}X - \hat{X}\right)^2\right]}
+\sqrt{D-(\sigma_X-\sigma_Y)^2}.
\end{align}
Since $\hat{X}$ has density and the Gaussian distribution maximizes differential entropy for a given variance (cf. \cite{CoverThomas_InfoTheoryBook}), we have the inequality $h(\hat{X})\leq \frac{1}{2}\log\left(2\pi e (\rho^2\sigma_X\sigma_Y+\sigma_Z^2)\right)$.  It follows that
\begin{align*}
I(X;\hat{X}) \leq \frac{1}{2}\log \left(\frac{\rho^2\sigma_X\sigma_Y+\sigma_Z^2}{\sigma_Z^2}\right) &= \log \left(\frac{2\sigma_X\sigma_Y}{\sigma_X^2+\sigma_Y^2-D} \right) \\
&= R_\ID(D,N(0,\sigma_X^2),N(0,\sigma_Y^2)).
\end{align*}
Thus, for $D\neq (\sigma_X-\sigma_Y)^2$, an application of Theorem \ref{thm:achGeneral} implies that
\begin{align}
 R_\ID(D,P_X,P_Y) \leq  R_\ID(D,N(0,\sigma_X^2),N(0,\sigma_Y^2)).\label{eqn:allbutOnePoint}
\end{align}
To handle the case where $D= (\sigma_X-\sigma_Y)^2$, we note that $R_\ID(D,P_X,P_Y)$ is nondecreasing in $D$.  Since
\begin{align}
\lim_{D\downarrow (\sigma_X-\sigma_Y)^2}R_\ID(D,N(0,\sigma_X^2),N(0,\sigma_Y^2)) = 0,
\end{align}
inequality \eqref{eqn:allbutOnePoint} implies that we must have $R_\ID(D,P_X,P_Y)=0$ at $D= (\sigma_X-\sigma_Y)^2$.  This completes the proof.
\end{proof}

\subsection{Robust Identification Schemes} \label{subsec:RobustProof}
Fix $R > R_\ID\left(D,N(0,\sigma_X^2),N(0,\sigma_Y^2)\right)$ and consider the setup described in section \ref{subsec:RobustSchemes}.  Specifically, let $P_{\tilde X}$, $P_{ \tilde Y}$ be zero-mean distributions with variances $\sigma_X^2$ and $\sigma_Y^2$, respectively.   Recall that, for a given blocklength $n$, the argument in the achievability proof of Theorem \ref{thm:RIDdiffVar} constructs a signature assignment function $T^{(n)} : \tilde \bx \ra \Reals^n$ for which the query $g^{(n)}\left(T^{(n)}(\tilde \bx),\tilde \by \right)$ returns ``$\MAYBE$" only if
\begin{enumerate}
\item The angle $\angle(\tilde \by,T^{(n)}(\tilde \bx))$ is at most $\theta'$,
where $\theta'<\pi/2$ is a fixed constant defined in \eqref{eqn:thetaprime}, and
\item We have $\tilde \bx\in  \ST_X$, where $\ST_X$ is the ``typical shell" of $\tilde \bx$ vectors defined in \eqref{sTypXdefn}.
\end{enumerate}
We remark that the gap between $\pi/2-\theta'$ and the thickness of the shell $\ST_X$ depend on the parameter $\eta>0$, which is a function of the gap between $R$ and $R_\ID\left(D,N(0,\sigma_X^2),N(0,\sigma_Y^2)\right)$.

  In light of the conditions under which $g^{(n)}\left(T^{(n)}(\tilde \bx),\tilde \by \right)$ returns ``$\MAYBE$", the probability of the event $\left\{g^{(n)}\left(T^{(n)}(\tilde \bX),\tilde \bY \right)=\MAYBE\right\}$ is bounded by
\begin{align}
\Pr\left\{g^{(n)}\left(T^{(n)}(\tilde \bX),\tilde \bY \right) = \MAYBE\right\} \leq \Pr \left\{\angle(\tilde \bY,T^{(n)}(\tilde \bX)) \leq \theta'  \right\} + \Pr\left\{\tilde \bX \notin \ST_X \right\}.
\end{align}
The term $\Pr\left\{\tilde \bX \notin \ST_X \right\}$ vanishes by the weak law of large numbers as $n\ra \infty$.  Therefore, since $\tilde \bX$ and $\tilde \bY$ are independent,  it is sufficient to show that $\Pr \left\{\angle(\tilde \bY,\bm{\alpha}) \leq \theta'  \right\}$ vanishes for any given unit vector $\bm{\alpha}=(\alpha_1,\alpha_2,\dots,\alpha_n)$ and constant $\theta'\in(0,\pi/2)$.  To this end, define $\beta_n \triangleq \frac{\sigma_Y}{2}\sqrt{n}$, and observe that
\begin{align}
\Pr \left\{\angle(\tilde \bY,\bm{\alpha}) \leq \theta'  \right\} &= \Pr \left\{ \sum_{i=1}^n \alpha_i \tilde  Y_i \geq \|\tilde  \bY\|\cos \theta' \right\} \\
&\leq \Pr \left\{ \sum_{i=1}^n \alpha_i \tilde  Y_i \geq \beta_n \cos \theta' \right\} +
\Pr \left\{ \|\tilde  \bY\| \leq \beta_n \right\}.
\end{align}
First, note $\lim_{n\rightarrow \infty } \Pr \left\{ \|\tilde \bY\| \leq \beta_n \right\}=0$ by the weak law of large numbers.
Next, since $\bm{\alpha}$ is a unit vector, we have $\sum_{i=1}^n \alpha_i^2=1$, and it follows that
\begin{align}
\VAR\left( \sum_{i=1}^n \alpha_i \tilde Y_i \right) =\sigma_Y^2.
\end{align}
Since $\EE[\tilde  Y_i]=0$, Chebyshev's inequality implies
\begin{align}
\Pr \left\{ \sum_{i=1}^n \alpha_i \tilde Y_i \geq \beta_n \cos \theta' \right\} \leq \frac{\sigma_Y^2}{\beta_n^2 \cos^2\theta'} = \frac{4}{n \cos^2\theta'},
\end{align}
proving that $\Pr\left\{g^{(n)}\left(T^{(n)}(\tilde  \bX),\tilde  \bY \right) = \MAYBE\right\} \ra 0$ as desired.  Since the $D$-admissibility of the scheme $(T^{(n)},g^{(n)})$ did not depend on the Gaussianity assumption in the proof of Theorem \ref{thm:RIDdiffVar}, the scheme $(T^{(n)},g^{(n)})$ continues to be $D$-admissible for the sources $\tilde\bX,\tilde\bY$.

Therefore, we can conclude that a sequence of rate-$R$, $D$-admissible schemes  $\{T^{(n)},g^{(n)}\}_{n\ra \infty}$ constructed as described in the proof of Theorem \ref{thm:RIDdiffVar} exhibit the robustness property explained in Section \ref{subsec:RobustSchemes}.

\section{Concluding Remarks}\label{sec:summary}
We studied the problem of answering similarity queries from compressed data from an information-theoretic perspective.
We focused on the setting where the similarity criterion is the (normalized) quadratic distance. For the case of i.i.d. Gaussian data,
we gave an explicit characterization of the minimal compression rate which permits reliable queries (i.e., the identification rate).
Furthermore, we characterized the best exponential rate at which the probability for false positives can be made to vanish.

For general sources, we derived an upper bound on the identification rate, and proved that it is at most that of the Gaussian source of the same variance.
Finally, we presented a single, robust, scheme that compresses \emph{any} source at the Gaussian identification rate, while permitting reliable responses to queries.
\section*{Acknowledgement}

The authors would like to thank Golan Yona for stimulating discussions that motivated this work.

\appendices
\section{Covering a Shell with Spheres}\label{app:ShellCovering}
\begin{proof}[Proof of Lemma~\ref{lem:ShellCovering}]
According to \cite[Theorem 1]{Dumer07}, for any $r>\rho>0$ there exists a covering of $S_r$ with balls of radius $\rho$, with density $\vartheta$ upper bounded by
\begin{align}
  \vartheta
  &\leq (n-1) \log(n-1)\left(\frac{1}{2} + \frac{2 \log \log (n-1) + 5}{\log (n-1)} \right)\label{eqn:halfnlogn}\\
  &\leq n \log(n)\label{eqn:nlogn},
\end{align}
where \eqref{eqn:halfnlogn} holds for all $n\geq 4$, and \eqref{eqn:nlogn} holds for $n$ large enough so that $\frac{2 \log \log (n-1) + 5}{\log (n-1)} \leq \frac{1}{2}$.
This translates to $k$ balls of radius $\rho$ that cover $S_r$, where
\begin{equation}\label{eqn:CoverSize1}
  k \leq \frac{n \log(n)}{\Omega(\theta)},
\end{equation}
and $\theta \triangleq \arcsin(\rho/r)$.

We choose $r = r_0 = \sqrt{n\sigma^2}$ and $\rho=\rho_0 = \sqrt{n D_0}$, so $\theta = \theta_0 = \arcsin(\sqrt{D_0/\sigma^2}) < \pi/2$ and is independent of $n$. When $n$ is large enough s.t. $\theta \leq \arccos(1/\sqrt n)$, we may use \eqref{eqn:OmegaLower} and get an upper bound on the covering size:
\begin{align}
  k &\leq \frac{n \log(n)}{\Omega(\theta_0)} \\
  &\leq \frac{n \log(n)}{\frac{1}{3\sqrt{2\pi n}\cos\theta_0} \sin^{n-1}\theta_0}\\
  &\leq 3\sqrt{2\pi}n^{3/2} \log(n) (\rho_0/r_0)^{n-1},
\end{align}
which proves \eqref{eqn:ShellCoveringExistence}.

Note that for a code that covers a spherical shell, the biggest covering by any single point $\bu\in\cC$ is obtained when the point $\bu$ is located at distance $\sqrt{r_0^2-\rho_0^2}$ from the origin. We therefore can assume, without altering the covering property of $\cC$, that $\|\bu\|=\sqrt{r_0^2-\rho_0^2}$ for all $\bu \in \cC$ (see also \cite[Eq. (1)]{Dumer07} and the discussion that follows). The intersection of $\BALL_{\rho_0}(\bu)$ and $S_{r_0}$ is precisely given by $\CAP_{r_0}(\bu,\theta_0)$.
\end{proof}

\section{}\label{app:Dexpansion}
\begin{proof}[Proof of Lemma~\ref{lem:Dexpansion}]
        Let $\by \in \Gamma^D\left(T^{-1}(\bu)\right) \cap \ST_Y$. Our goal is to show that $\by \in \CONE(\bu,\theta')$. In other words, we need to show that
        \begin{equation}
          \angle(\bu,\by) \leq \theta'.
        \end{equation}

        Since $\by \in \Gamma^D T^{-1}(\bu)$, there exists $\bx \in T^{-1}(\bu)$ s.t. $d(\bx,\by) \leq D$. By the triangle inequality for the angle operator (which is proportional to the geodesic metric in spherical geometry) we can write
        \begin{align}
          \angle(\bu,\by) &\leq \angle(\bu,\bx)+\angle(\bx,\by).
        \end{align}
        Since $T^{-1}(\bu)\subseteq \CAP_{r^-,r^+}(\bu,\theta_0)$, we know that $\angle(\bu,\bx) \leq \theta_0$. Further, by the law of cosines for the triangle $(\bx,\by,\mathbf{0})$ we can write
        \begin{align}
            \angle(\bx,\by)
            &= \arccos\left[\frac{\|\bx\|^2 + \|\by\|^2- \|\bx-\by\|^2}{2\|\bx\|\|\by\|}\right]\\
            &\overset{(a)}\leq \arccos\left[\frac{\sigma_X^2+ \sigma_Y^2-2\eta- nD}{2
            \sqrt{(\sigma_X^2+\eta)(\sigma_Y^2+\eta)}}
            \right]\\
            & = \theta_1,
        \end{align}
        where $(a)$ follows since $\bx \in \ST_X$, $\by \in \ST_Y$ and $d(\bx,\by) \leq D$. Therefore by definition we have $\by \in \CONE(\bu,\theta')$. All there's left to show is that $\theta' < \frac{\pi}{2}$. This follows immediately since $D_0$ satisfies \eqref{eqn:DefnD0} by definition, and from the fact that $\arcsin(\phi) + \arccos(\phi) = \frac{\pi}{2}$.
\end{proof}

\section{}\label{app:Markov}
\begin{proof}[Proof of Lemma~\ref{lem:Markov}]
Define $\cI$ as the set of indices $i$ for which $p_i \geq p^*$:
\begin{equation}
  \cI \triangleq \{i : p_i \geq p^*\}.
\end{equation}
Clearly $\Omega\left(\theta_{D''}+\Omega^{-1}(p_i)\right) \geq \Omega^*$ if and only if $i\in \cI$, so $\cI$ can be thought of as the set of `bad' values for $i$, i.e. those that contribute a lot to the sum in \eqref{eqn:Sum_pi_Omega}.

Consider the following sequence of inequalities:
\begin{align}
  c\cdot \Omega^*
  &\geq \sum_{i=1}^{2^{nR}} p_i \cdot \Omega\left(\theta_{D''}+\Omega^{-1}(p_i)\right)  \nonumber\\
  &\geq \sum_{i\in \cI}p_i \cdot \Omega\left(\theta_{D''}+\Omega^{-1}(p_i)\right)  \nonumber\\
  &\geq \Omega^* \sum_{i\in \cI}p_i. \label{eqn:bound_i_in_cI}
\end{align}

On the other hand,
\begin{align*}
  1
  &= \sum_i p_i \\
  &= \sum_{i\in \cI} p_i + \sum_{i\notin \cI} p_i\\
  &\overset{(a)}\leq c + \sum_{i\notin \cI} p_i\\
  &\overset{(b)}\leq c + \sum_{i\notin \cI} p^*\\
  &\leq c + p^*2^{nR}.
\end{align*}
$(a)$ follows from \eqref{eqn:bound_i_in_cI}. $(b)$ follows from the definition of $\cI$. Eq. \eqref{eqn:Markov} follows immediately.
\end{proof}

\section{}\label{app:supex}

\begin{proof}[Proof of Lemma~\ref{lem:supex}]
For any $t>0$ and $a>0$ we have
\begin{align*}
  \Pr\{\|\bX\|^2 > a \}
  &\overset{(a)}\leq e^{-t \cdot a} \EE\left[\exp\left(t \cdot \|\bX\|^2\right)\right]\\
  &\overset{(b)}=  e^{-t \cdot a} (1-2t\sigma_X^2)^{-n/2}.
\end{align*}
$(a)$ is the Chernoff bound. $(b)$ follows since the moment generating function of $Z$ is given by
\begin{equation}\label{eqn:Zmgf}
  \EE\left[e^{tZ}\right] = (1-2t)^{-n/2}, \mbox{ for $t<\frac{1}{2}$.}
\end{equation}
Here it holds for any $t<\frac{1}{2\sigma_X^2}$. We choose $t=\frac{1}{4\sigma_X^2}$ and write:
\begin{equation}
  \Pr\{\|\bX\|^2 > a \}  \leq  e^{- \frac{a}{4\sigma_X^2}}\cdot  2^{n/2}.
\end{equation}
Choosing $a=n\sigma_{\max}^2(n) = n^2\sigma_X^2$ results in
\begin{equation}
  \Pr\{\|\bX\|^2 > n\sigma_{\max}^2(n) \}  \leq  e^{- \tfrac{1}{4}n^2 + o(n^2)}.
\end{equation}

\end{proof}

\bibliographystyle{IEEEtran}
\bibliography{Master}

\begin{thebibliography}{10}
\providecommand{\url}[1]{#1}
\csname url@samestyle\endcsname
\providecommand{\newblock}{\relax}
\providecommand{\bibinfo}[2]{#2}
\providecommand{\BIBentrySTDinterwordspacing}{\spaceskip=0pt\relax}
\providecommand{\BIBentryALTinterwordstretchfactor}{4}
\providecommand{\BIBentryALTinterwordspacing}{\spaceskip=\fontdimen2\font plus
\BIBentryALTinterwordstretchfactor\fontdimen3\font minus
  \fontdimen4\font\relax}
\providecommand{\BIBforeignlanguage}[2]{{%
\expandafter\ifx\csname l@#1\endcsname\relax
\typeout{** WARNING: IEEEtran.bst: No hyphenation pattern has been}%
\typeout{** loaded for the language `#1'. Using the pattern for}%
\typeout{** the default language instead.}%
\else
\language=\csname l@#1\endcsname
\fi
#2}}
\providecommand{\BIBdecl}{\relax}
\BIBdecl

\bibitem{Ahlswede97}
R.~Ahlswede, E.-h. Yang, and Z.~Zhang, ``Identification via compressed data,''
  \emph{Information Theory, IEEE Transactions on}, vol.~43, no.~1, pp. 48 --70,
  Jan 1997.

\bibitem{tuncel2004rate}
E.~Tuncel, P.~Koulgi, and K.~Rose, ``Rate-distortion approach to databases:
  Storage and content-based retrieval,'' \emph{IEEE Trans. on Information
  Theory}, vol.~50, no.~6, pp. 953--967, 2004.

\bibitem{WynerZiv}
A.~Wyner and J.~Ziv, ``The rate-distortion function for source coding with side
  information at the decoder,'' \emph{IEEE Trans. on Information Theory},
  vol.~22, no.~1, pp. 1 -- 10, jan 1976.

\bibitem{OSullivan2002}
J.~O'Sullivan and N.A.Schmid, ``Large deviations performance analysis for
  biometrics recognition,'' in \emph{Allerton Conference on Communication,
  Control, and Computing}, October 2002.

\bibitem{Willems03}
F.~Willems, T.~Kalker, S.~Baggen, and J.~paul Linnartz, ``On the capacity of a
  biometrical identification system,'' in \emph{In: Proc. of the 2003 IEEE Int.
  Symp. on Inf. Theory}, 2003, pp. 8--2.

\bibitem{Westover08}
M.~Westover and J.~O'Sullivan, ``Achievable rates for pattern recognition,''
  \emph{Information Theory, IEEE Transactions on}, vol.~54, no.~1, pp. 299
  --320, jan. 2008.

\bibitem{Tuncel09}
E.~Tuncel, ``Capacity/storage tradeoff in high-dimensional identification
  systems,'' \emph{Information Theory, IEEE Transactions on}, vol.~55, no.~5,
  pp. 2097 --2106, may 2009.

\bibitem{Tuncel12submitted}
E.~Tuncel and D.~G\"und\"uz, ``Identification and lossy reconstruction in noisy
  databases,'' \emph{Submitted to IEEE Transactions on Information Theory},
  2012.

\bibitem{Bloom1970}
B.~H. Bloom, ``Space/time trade-offs in hash coding with allowable errors,''
  \emph{Commun. ACM}, vol.~13, no.~7, pp. 422--426, Jul. 1970.

\bibitem{Porat09_Opt_Bloom_Filter_matrix}
E.~Porat, ``An optimal bloom filter replacement based on matrix solving,'' in
  \emph{CSR}, ser. Lecture Notes in Computer Science, A.~E. Frid, A.~Morozov,
  A.~Rybalchenko, and K.~W. Wagner, Eds., vol. 5675.\hskip 1em plus 0.5em minus
  0.4em\relax Springer, 2009, pp. 263--273.

\bibitem{MitzenmacherBloomFilter}
A.~Z. Broder and M.~Mitzenmacher, ``Survey: Network applications of {B}loom
  filters: A survey,'' \emph{Internet Mathematics}, vol.~1, no.~4, pp.
  485--509, 2003.

\bibitem{AndoniI08}
A.~Andoni and P.~Indyk, ``Near-optimal hashing algorithms for approximate
  nearest neighbor in high dimensions,'' \emph{Commun. ACM}, vol.~51, no.~1,
  pp. 117--122, 2008.

\bibitem{JohnsonLindenstrauss84}
W.~B. Johnson and J.~Lindenstrauss, ``Extensions of lipschitz mappings into a
  {H}ilbert space,'' \emph{Conf. in Modern Analysis and Probability}, vol.~26,
  pp. 189--206, 1984, conf. was held in 1982, book publ. 1984.

\bibitem{sketchingNotes}
P.~Indyk, \emph{Sketching, streaming and sublinear-space algorithms}.\hskip 1em
  plus 0.5em minus 0.4em\relax Lecture Notes, 2007, Mass. Inst. of Tech.,
  available at http://stellar.mit.edu/S/course/6/fa07/6.895/.

\bibitem{BR_DCC13}
P.~T. Boufounos and S.~Rane, ``Efficient coding of signal distances using
  universal quantized embeddings,'' in \emph{Proc. Data Compression Conference
  (DCC)}, Snowbird, UT, March 20-22 2013.

\bibitem{GallagerInfoTheoryBook}
R.~G. Gallager, \emph{Information Theory and Reliable Communication}.\hskip 1em
  plus 0.5em minus 0.4em\relax New York, NY, USA: John Wiley \& Sons, Inc.,
  1968.

\bibitem{CoverThomas_InfoTheoryBook}
T.~M. Cover and J.~A. Thomas, \emph{Elements of Information Theory},
  2nd~ed.\hskip 1em plus 0.5em minus 0.4em\relax John Wiley \& Sons, 2006.

\bibitem{Marton1974fidelityCriterion}
K.~Marton, ``Error exponent for source coding with a fidelity criterion,''
  \emph{IEEE Trans. on Information Theory}, vol.~20, no.~2, pp. 197 -- 199, mar
  1974.

\bibitem{IharaKubo2000}
S.~Ihara and M.~Kubo, ``Error exponent of coding for memoryless {G}aussian
  sources with a fidelity criterion,'' \emph{IEICE Trans. on Fundam. Electron.
  Commun. Comput. Sci.}, vol. 83-A, no.~10, pp. 1891--1897, 2000.

\bibitem{IharaKubo2005}
------, ``Error exponent of coding for stationary memoryless sources with a
  fidelity criterion,'' \emph{IEICE Trans. on Fundam. Electron. Commun. Comput.
  Sci.}, vol. E88-A, no.~5, pp. 1339--1345, May 2005.

\bibitem{zhongAC2006Laplacian}
Y.~Zhong, F.~Alajaji, and L.~L. Campbell, ``A type covering lemma and the
  excess distortion exponent for coding memoryless {L}aplacian sources,'' in
  \emph{23rd Biennial Symposium on Communications}.\hskip 1em plus 0.5em minus
  0.4em\relax IEEE, 2006, pp. 100--103.

\bibitem{DemboZeitouni}
A.~Dembo and O.~Zeitouni, \emph{Large Deviations Techniques and Applications},
  2nd~ed., ser. Stochastic Modelling and Applied Probability.\hskip 1em plus
  0.5em minus 0.4em\relax Springer, 1998, vol.~38.

\bibitem{boroczky03}
K.~B\"{o}r\"{o}czky~Jr. and G.~Wintsche, ``Covering the sphere by equal
  spherical balls,'' in \emph{Discrete and Computational Geometry: The
  Goodman-Pollack Festschrift}.\hskip 1em plus 0.5em minus 0.4em\relax
  Springer, 2003, pp. 237--253.

\bibitem{Dumer07}
I.~Dumer, ``Covering spheres with spheres.'' \emph{Discrete \& Computational
  Geometry}, vol.~38, no.~4, pp. 665--679, 2007.

\bibitem{ledoux2011probability}
M.~Ledoux and M.~Talagrand, \emph{Probability in Banach Spaces: Isoperimetry
  and Processes}, ser. Ergebnisse der Mathematik Und Ihrer Grenzgebiete.\hskip
  1em plus 0.5em minus 0.4em\relax Springer, 2011.

\bibitem{GershoGray}
A.~Gersho and R.~M. Gray, \emph{Vector Quantization and Signal
  Compression}.\hskip 1em plus 0.5em minus 0.4em\relax Kluwer, 1992.

\bibitem{ElGamalYHKim2012}
A.~El~Gamal and Y.~Kim, \emph{Network Information Theory}.\hskip 1em plus 0.5em
  minus 0.4em\relax Cambridge University Press, 2011.

\end{thebibliography}

\end{document}